      \documentclass[ 10pt, numbers, nocopyrightspace]{sigplanconf}

 \usepackage{balance} 

\title{Achieving High Coverage for Floating-point Code\\ via Unconstrained Programming (Extended Version)} 

 \authorinfo{Zhoulai Fu \qquad Zhendong Su}{University of California, Davis, USA}{zhoulai.fu@gmail.com \qquad su@ucdavis.edu}

\makeatletter
\newcommand{\removelatexerror}{\let\@latex@error\@gobble}
\makeatother
\usepackage{todonotes}

\usepackage[all]{xy}

\usepackage{siunitx}
\usepackage{epigraph}

\usepackage[scaled=0.81]{beramono} % make tt fonts smaller (default scaled : 0.9)
\usepackage{enumitem}
\usepackage{graphicx}
\usepackage{xcolor}
\usepackage{amsmath}
\usepackage{amsfonts}
\usepackage{amssymb}
\usepackage{amsthm}
\usepackage[hyphens]{url} % do I need this with hyperref?
\usepackage{multicol}
\usepackage{multirow}
\usepackage{xspace}

\usepackage{ifthen}

\usepackage{float} % defines newfloat used below
\usepackage{microtype}
\usepackage{mathptmx}
\usepackage{stmaryrd} %Additional math symbols.
\usepackage{textcomp}
\usepackage{fixltx2e} 
\usepackage{booktabs}
\usepackage{colortbl} %color the tabular horizontal and vertical rules
\usepackage{tabularx}
\usepackage{fancyhdr}
\usepackage[mathletters]{ucs}
\usepackage[utf8x]{inputenc}
\usepackage{wasysym}	% Defines \cent, \currency, \brokenvert
\usepackage[normalem]{ulem}	% adds \sout for strikeout
\usepackage{balance} 
\usepackage[font=small]{caption}
\usepackage{subcaption}
\usepackage[T1]{fontenc}
\usepackage{fixltx2e} 
\usepackage{framed}
\usepackage[linesnumbered,ruled,vlined]{algorithm2e}
\SetCommentSty{mycommfont}

\usepackage[%
	bookmarks, colorlinks, citecolor=blue, urlcolor=black, %draft,
	filecolor=black, linkcolor=blue  % pdfpagelabels=false%
]{hyperref}

\usepackage{listings}
\usepackage{bbm}
    \usepackage[bottom]{footmisc}
\usepackage{bbold}

\newcommand{\etal}{\hbox{\emph{et al.}}\xspace}
\newcommand{\eg}{\hbox{\emph{e.g.}}\xspace}
\newcommand{\ie}{\hbox{\emph{i.e.}}\xspace}
\newcommand{\etc}{\hbox{\emph{etc.}}\xspace}
\newcommand{\aka}{\hbox{{a.k.a.}}\xspace} %videlicet: it is permitted to see

 \newcommand{\FP}{\mathbbm{F}}
\newcommand{\mydef}{\overset{\text{def}}{=}}

\newtheorem{thm}{Theorem}[section]
\newtheorem{lem}[thm]{Lemma}

\theoremstyle{definition}
 
\newtheorem{definition}[thm]{Definition}

 \newtheorem{remark}[thm]{Remark}

\newcommand{\Real}{\mathbbm{R}}

\newcommand{\energy}{f}

\SetKwFunction{RF}{RF}
\SetKwFunction{proc}{Basin-Hopping}
\SetKwFunction{procshort}{MH}
\definecolor{light-gray}{gray}{0.8}

\newcommand{\dom}[1]{\operatorname{dom}(#1)\xspace}

\newcommand{\FOO}{{\texttt {FOO}}\xspace}

\newcommand{\FOOI}{{\texttt {FOO\_I}}\xspace}
\newcommand{\FOOR}{{\texttt {FOO\_R}}\xspace}
\newcommand{\loader}{{\texttt {loader}}\xspace}
\newcommand{\librso}{{\texttt {libr.so}}\xspace}
\newcommand{\mcmc}{{\texttt {MCMC}}\xspace}
\newcommand{\pen}{{\mathnormal{pen}}\xspace}
 \newcommand{\myr}{\texttt {r}\xspace}

\newcommand{\about}[1]{\paragraph{}\noindent{{\underline{\sc  #1}}}}
\renewcommand{\about}[1]{}

\newcommand{\niter}{\mathnormal {n\_iter}\xspace}
\newcommand{\nStart}{\mathnormal {n\_start}\xspace}
\newcommand{\LM}{\texttt {LM}\xspace}
\newcommand{\lmin}{\texttt {LM}\xspace}

\newcommand{\dist}{\mathnormal{d}}

\usepackage{tikz}

\newcommand{\Integer}{\mathbbm{Z}}

\newcommand{\slug}{\hbox{\kern1.5pt\vrule width2.5pt height6pt depth1.5pt\kern1.5pt}}
\def\xskip{\hskip 7pt plus 3pt minus 4pt}
\newdimen\algindent
\newif\ifitempar \itempartrue % normally true unless briefly set false
\def\algindentset#1{\setbox0\hbox{{\bf #1.\kern.25em}}\algindent=\wd0\relax}
\def\algbegin #1 #2{\algindentset{#21}\alg #1 #2} % when steps all have 1 digit
\def\aalgbegin #1 #2{\algindentset{#211}\alg #1 #2} % when 10 or more steps
\def\alg#1(#2). {\medbreak % Usage: \algbegin Algorithm A (algname). This...
  \noindent{\bf#1}({\it#2\/}).\xskip\ignorespaces}
\def\algstep#1.{\ifitempar\smallskip\noindent\else\itempartrue
  \hskip-\parindent\fi
  \hbox to\algindent{\bf\hfil #1.\kern.25em}%
  \hangindent=\algindent\hangafter=1\ignorespaces}

\newcommand{\coverme}{{\textnormal{CoverMe}}\xspace}

\newcommand{\fdlibm}{\textnormal{Fdlibm}}
\DeclareMathOperator{\Explored}{Saturate}

\newcommand{\myT}{_{\mathnormal T}}
\newcommand{\myF}{_{\mathnormal F}}
\newcommand{\mylhs}{\mathnormal{a}}
\newcommand{\myrhs}{\mathnormal{b}}
\newcommand{\myop}{\mathnormal{op}}

\newcommand{\myX}{\mathnormal{X}}

\newcommand{\Paragraph}[1]{\paragraph{\upshape #1}}

\newcommand{\xloc}{x_{\textnormal{L}}}
\newcommand{\xpro}{\widetilde{\xloc}}

\begin{document}

\maketitle

%%%%%
% git add intro.tex overview.tex algo.tex  implem.tex eval.tex relwork.tex conc.tex abstract.tex macros.tex sigplanconf.cls  main.tex main.bib
% git add fig/*.png
%

 % %%%%%%%%%%%%To include

\begin{abstract}

Achieving high code coverage is essential in testing, which gives us
confidence in code quality. Testing floating-point code usually
requires painstaking efforts in handling floating-point constraints,
\eg, in symbolic execution. This paper turns the challenge of testing
floating-point code into the opportunity of applying unconstrained
programming --- the mathematical solution for calculating function minimum
points over the entire search space. Our core insight is to derive
 a representing function from the floating-point program, any of whose minimum
points is a test input guaranteed to exercise a new branch of the
tested program. This guarantee allows us to achieve high coverage of
the floating-point program by repeatedly minimizing the representing
function.

We have realized this approach in a tool called CoverMe and conducted
an extensive evaluation of it on Sun's C math library. Our evaluation results show
that CoverMe achieves, on average, 90.8\% branch coverage in 6.9 seconds, drastically
outperforming our compared tools: (1) Random testing, (2) AFL, a highly
optimized, robust fuzzer released by Google, and (3) Austin, a
state-of-the-art coverage-based testing tool designed to support
floating-point code.

\end{abstract}

\section{Introduction}
\label{sect:intro}
Test coverage criteria attempt to quantify the quality of test
data. Coverage-based testing~\cite{DBLP:dblp_books/daglib/0012071} has
become the state-of-the-practice in the software industry. The
higher expectation for software quality and the shrinking development cycle
have driven the research community to develop a spectrum of
automated testing techniques for achieving high code coverage.

%are forcing software engineers to seek
%help from automation --- to automatically generate test data for
%achieving high code coverage.

A significant challenge in coverage-based testing lies in the testing
of numerical code, \eg, programs with floating-point arithmetic,
non-linear variable relations, or external function calls, such as
logarithmic and trigonometric functions.  Existing solutions include
random
testing~\cite{Bird:1983:AGR:1662311.1662317,DBLP:conf/pldi/GodefroidKS05},
symbolic
execution~\cite{King:1976:SEP:360248.360252,Boyer:1975:SFS:800027.808445,Clarke:1976:SGT:1313320.1313532,DBLP:journals/cacm/CadarS13},
and various search-based
strategies~\cite{Korel:1990:AST:101747.101755,McMinn:2004:SST:1077276.1077279,Baars:2011:SST:2190078.2190152,Lakhotia:2010:FSF:1928028.1928039},
which have found their way into many mature
implementations~\cite{Cadar:2008:KUA:1855741.1855756,DBLP:conf/kbse/BurnimS08,Sen:2005:CCU:1081706.1081750}.
Random testing is easy to employ and fast, but ineffective in finding
deep semantic issues and handling large input spaces; symbolic
execution and its variants can perform systematic path exploration,
but suffer from path explosion and are weak in dealing with complex
program logic involving numerical constraints.
%DBLP:dblp_conf/sigsoft/SenMA05,
\Paragraph{Our Work}
This paper considers the problem of coverage-based testing for
floating-point code and focuses on the coverage of program
branches.  We turn the challenge of testing floating-point programs
into the opportunity of applying unconstrained programming ---
the mathematical solution for calculating function minima  over the
entire search space~\cite{Zoutendijk76,Minoux86}. % (Sect.~\ref{sect:background}).

Our approach has two unique features. First, it introduces the concept
of  \emph{representing function}, which reduces the branch coverage
based testing problem to the unconstrained programming problem.
Second, the representing function is specially designed to achieve the following
theoretical guarantee: Each minimum point of the representing function
is an input to the tested floating-point program, and the input
necessarily triggers a new branch unless all branches have been
covered.  This guarantee is critical not only for the soundness of our
approach, but also for its efficiency --- the unconstrained programming
process is designed to cover only new branches; it does not waste
efforts on covering already covered branches.

%the unconstrained programming process is
%dedicated to  covering new branches; it spares no efforts  covering the
%branches that the previously generated test data have already  covered.

% To clarify, consider a program under test \FOO.  Our approach derives from \FOO  another program \FOOR so that an input $x^*\in\dom{\FOO}$  

% Let $\FOO$ be the program under test. We transform the problem of finding the input set that covers all branches of $\FOO$  to the problem of finding all minimum points of another floating-point program $\FOOR$. We call this program $\FOOR$ the representing function of $\FOO$, a term we borrow from mathematics~\cite{Kleene62}.  Once we have $\FOOR$, we then handle the 

We have implemented our approach into a tool called CoverMe. CoverMe first
derives the representing function from the program under test. Then,
it uses an existing unconstrained programming algorithm to compute
the minimum points.  Note that the theoretical guarantee mentioned
above allows us to apply any unconstrained programming algorithm as a black
box. Our implementation uses an off-the-shelf Monte Carlo Markov Chain
(MCMC)~\cite{DBLP:dblp_journals/ml/AndrieuFDJ03} tool.

CoverMe has achieved high or full branch coverage for the tested
floating-point programs.  Fig.~\ref{fig:intro:fdlibm} lists the
program {\tt s\_tanh.c}  from our benchmark suite
$\fdlibm$~\cite{fdlibm:web}.  The program takes a {\tt double} input.
In Line 3, variable {\tt jx} is assigned with the high word of {\tt x}
according to the comment given in the source code; the right-hand-side
expression in the assignment takes the address of {\tt x} ({\tt \&x}),
cast it as a pointer-to-int ({\tt int*}), add 1, and dereference the
resulting pointer.  In Line 4, variable {\tt ix} is assigned with {\tt
  jx} whose sign bit is masked off.  Lines 5-15
are two nested conditional statements on {\tt ix} and {\tt jx}, which
contain 16 branches in total according to Gcov~\cite{gcov:web}.
Testing this type of programs is beyond the capabilities of traditional
symbolic execution tools such as
Klee~\cite{Cadar:2008:KUA:1855741.1855756}. CoverMe achieves full
coverage within 0.7 seconds, dramatically outperforming our
compared tools, including random testing, Google's AFL, and Austin
(a tool that combines symbolic execution and search-based
heuristics). See details in Sect.~\ref{sect:eval}.

% which largely outperform  We have compared CoverMe with
% random testing, AFL (a highly optimized and robust fuzzy released by Google), and Austin (a
% state-of-the-art coverage-based testing tool designed for supporting
% floating-point). Their branch coverage are 

% Austin~\cite{lakhotia2013austin}, a state-of-the-art coverage-based
% testing tool that has been well-tested on floating-point
% coverage-based testing~\cite{lakhotia2010empirical}.  Austin achieves
% 33.3\% branch coverage in 1885.1 seconds.  We have also compared
% CoverMe with a pure random testing and AFL~\cite{afl:web}, a
% state-of-the-art automated testing tool that employs genetic
% algorithms and search heuristics. Their coverage is 33.3\% and 75.0\%
% respectively (see Sect.~\ref{sect:eval} for details of our settings
% and results).

%#define __HI(x) *(1+(int*)&x)

\begin{figure}
\lstset{xleftmargin=0.5cm, numbers=left}
\begin{lstlisting}
double tanh(double x){
  int jx, ix;
  jx = *(1+(int*)&x);  // High word of x
  ix = jx&0x7fffffff; 
  if(ix>=0x7ff00000) { 
    if (jx>=0) ...;
    else       ...;
  }
  if (ix < 0x40360000) {	
    if (ix<0x3c800000) ...;
    
    if (ix>=0x3ff00000) ...; 	
    else ...;
  }
  else ...;
  return ...;
}

\end{lstlisting}
\caption{Benchmark program \texttt{s\_tanh.c} taken from $\fdlibm$. }\label{fig:intro:fdlibm}
\end{figure}

\Paragraph{Contributions} %Our contributions are two-fold:

This work introduces a promising automated testing solution for
programs that are heavy on floating-point computation. Our approach
designs the representing function whose minimum points are guaranteed
to exercise new branches of the floating-point program. This guarantee
allows us to apply any unconstrained programming solution as a
black box, and to efficiently generate test inputs for
covering program branches.

Our implementation, CoverMe, proves to be highly efficient and effective. It 
achieves 90.8\% branch coverage on average, which is substantially
higher than those obtained by random testing (38.0\%),
AFL~\cite{afl:web} (72.9\%), and Austin~\cite{lakhotia2013austin}
(42.8\%).

% --  three  tools
%available to us for floating-point coverage-based testing.

%  These are the best tools available to us for  floating-point coverage-based testing. They attain 48\%

 %Compared with random testing and
 % AFL, our approach achieves 52.9\% and 17.9\% more coverage
 % respectively. Compared with Austin our approach enhances the branch
 % coverage by 48\%, with significantly less time (6.9 \vs 6058.4
 % seconds) on average.

\Paragraph{Paper Outline}
We structure the rest of the paper as follows.
Sect.~\ref{sect:background} presents background material on
unconstrained programming.  Sect.~\ref{sect:overview} gives an
overview of our approach, and Sect.~\ref{sect:algo} presents the
algorithm. Sect.~\ref{sect:implem} describes our implementation
CoverMe, and Sect.~\ref{sect:eval} describes our evaluation.
Sect.~\ref{sect:relwork} surveys related work and
Sect.~\ref{sect:conc} concludes the paper. For completeness,
Sect.~\ref{sect:untested}-\ref{sect:incompleteness}  provide additional  details on our approach.

\Paragraph{Notation} We write $\FP$ for floating-point
numbers, $\Integer$ for integers, $\Integer_{>0}$ for strictly
positive integers.  we use the ternary operation $B~?~a~:~a'$ to denote
an evaluation to $a$ if $B$ holds, or $a'$ otherwise. The lambda terms in the form of $\lambda x.f(x)$ may denote mathematical function $f$ or its machine implementation according to the given context.

\section{Background}
\label{sect:background}

This section presents the definition and  algorithms of unconstrained
programming that will be used in this paper. 
 As mentioned in Sect.~\ref{sect:intro}, we will treat
the unconstrained programming algorithms as black boxes. 

% We begin
%with some notions on unconstrained optimization.

%Below, we write $\Real$ for the set
%of reals and $\Real^n$ for an $n$-dimensional space of reals.  

%The functions involved in this section   are
%real-valued in theory, but in practice, they are floating-point programs.
%As mentioned in Sect.~\ref{sect:intro}, these  algorithms will be treated as blackbox and we assume they are available to us. 

\Paragraph{Unconstrained Programming}
 We formalize  unconstrained programming  as the problem below~\cite{dennis1996numerical}:
\begin{equation*} \label{eq:background:1}
 \begin{aligned}
 & {\text{Given}}
 & & f:\Real^n\rightarrow \Real \\
 & \text{Find}
 & & x^*\in \Real^n \text{ for which } f(x^*)\leq f(x) \text{ for all } x\in\Real^n 
 \end{aligned}
 \end{equation*}
 where $f$ is  the objective function;  $x^*$, if found, is  called a
 minimum point; and $f(x^*)$ is the minimum. %Note that it is possible to have 
% more than one minimum points, but only  one minimum. 
 An example is
 \begin{align}
f(x_1,x_2)=(x_1-3)^2+(x_2-5)^2,   
 \end{align}
 which has the minimum point  $x^*=(3,5)$. % with minimum $0$. 

\Paragraph{Unconstrained Programming Algorithms}
We consider two kinds of algorithms, known as local optimization and
global optimization.  Local optimization focuses on how functions are
shaped near a given input and where a minimum can be found at local
regions.  It usually involves standard techniques such as Newton's or
the steepest descent methods~\cite{citeulike:2621649}.  Fig.~\ref{fig:background:1}(a)
shows a common local optimization method with the objective function
$f(x)$ that equals $0$ if $x\leq 1$, or $(x-1)^2$ otherwise.  The
algorithm uses tangents of $f$ to converge to a minimum point
quickly. In general, local optimization is usually fast.  If the
objective function is smooth to some degree, the local optimization
can deduce function behavior in the neighborhood of a particular point
$x$ by using information at $x$ only (the tangent here).

Global optimization for unconstrained programming searches for minimum
points over $\Real^n$.  Many global optimization algorithms have been
developed. This work uses  Monte Carlo Markov Chain
(MCMC)~\cite{DBLP:dblp_journals/ml/AndrieuFDJ03}. MCMC is a sampling
method that targets (usually unknown) probability distributions. A
fundamental fact is that MCMC sampling follows the target distributions
asymptotically, which is formalized by the lemma below. For
simplicity, we present the lemma in the form of discrete-valued
probabilities~\cite{DBLP:dblp_journals/ml/AndrieuFDJ03}.
 
\begin{lem}
  Let $x$ be a random variable, $A$ be an enumerable set of the
  possible values of $x$,   $f$ be a target probability distribution
  for $x$, \ie, the probability of $x$ taking value $a\in A$ is $f(a)$.
  Then,  an MCMC sampling sequence $x_1, \ldots, x_n,\ldots$ satisfies the property that
  $Prob(x_n =a) \rightarrow f(a)$.
\label{lem:pdf}
\end{lem}
For example, consider the target distribution of coin tossing with
$0.5$ probability of getting a head. An MCMC sampling is a sequence
of random variables $x_1$,\ldots, $x_n,\ldots$, such that the
probability of $x_n$ being head converges to $0.5$.

% Other unconstrained  optimization techniques, \eg, genetic
%programming~\cite{koza1992genetic}, may also be applicable, which we
%leave for future investigation.  

Using MCMC to solve unconstrained programming problems provides
multiple advantages in practice. First, Lem.~\ref{lem:pdf} ensures
that MCMC sampling can be configured to attain the minimum points with
higher probability than the other points.  Second, MCMC integrates
well with local optimization. An example is the Basinhopping
algorithm~\cite{leitner1997global} used in Sect.~\ref{sect:implem}.
Third, MCMC techniques are robust; some variants can even handle high
dimensional problems~\cite{robbins1951stochastic} or non-smooth
objective functions~\cite{eckstein1992douglas}. Our approach uses
unconstrained optimization as a black box.
Fig.~\ref{fig:background:1}(b) provides a simple example.  Steps
$p_0\rightarrow p_1$, $p_2\rightarrow p_3$, and $p_4 \rightarrow p_5$
employ local optimization; Steps $p_1\rightarrow p_2$ and
$p_3\rightarrow p_4$, known as Monte-Carlo
moves~\cite{DBLP:dblp_journals/ml/AndrieuFDJ03}, avoid the MCMC
sampling being trapped in the local minimum points.

 \begin{figure}[t]
\begin{subfigure}[b]{0.45\linewidth}
   \centering
  \includegraphics[width=\linewidth]{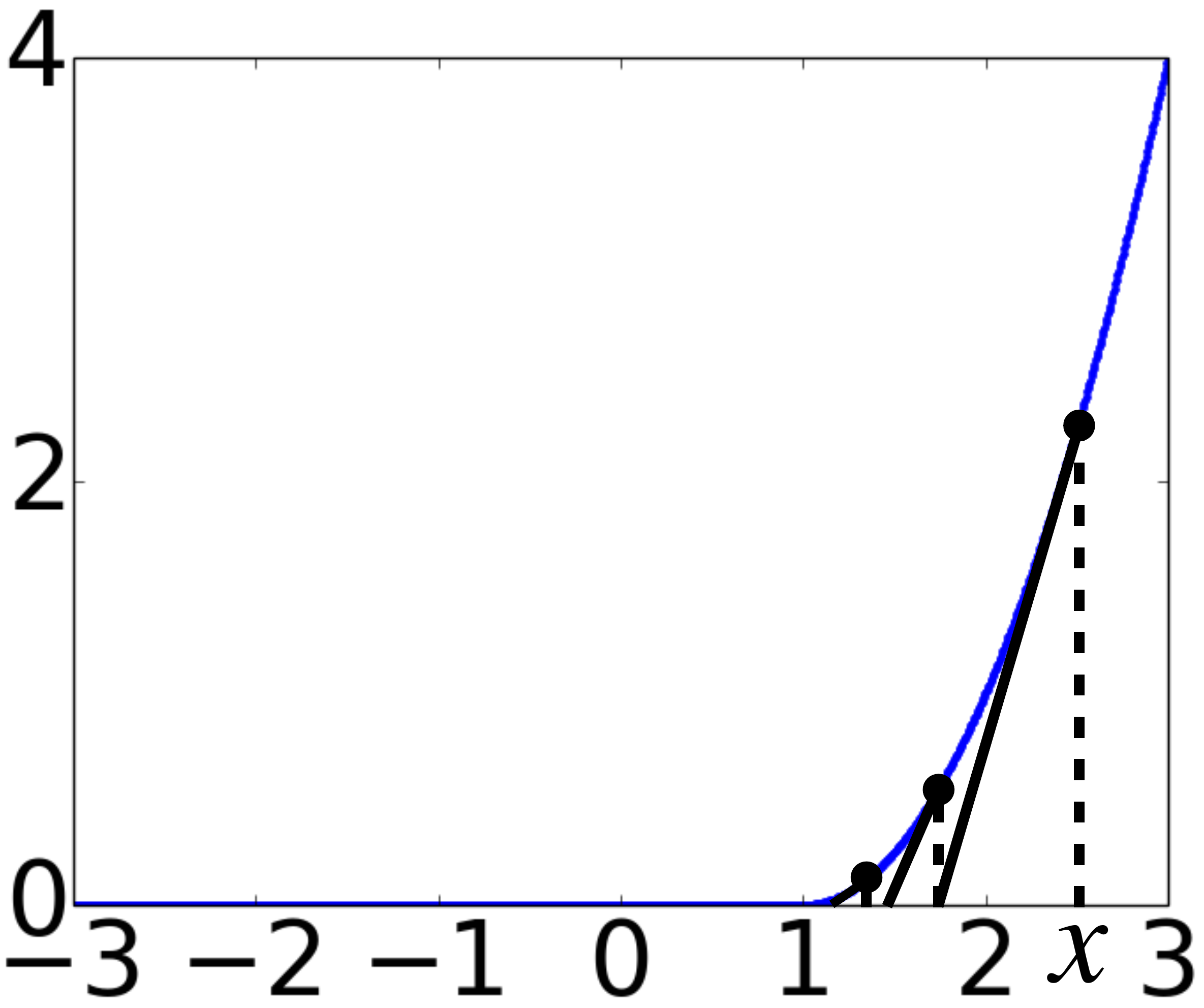}
  \caption{}
\end{subfigure}%
\hfill
\begin{subfigure}[b]{0.46\linewidth}
   \centering
  \includegraphics[width=\linewidth]{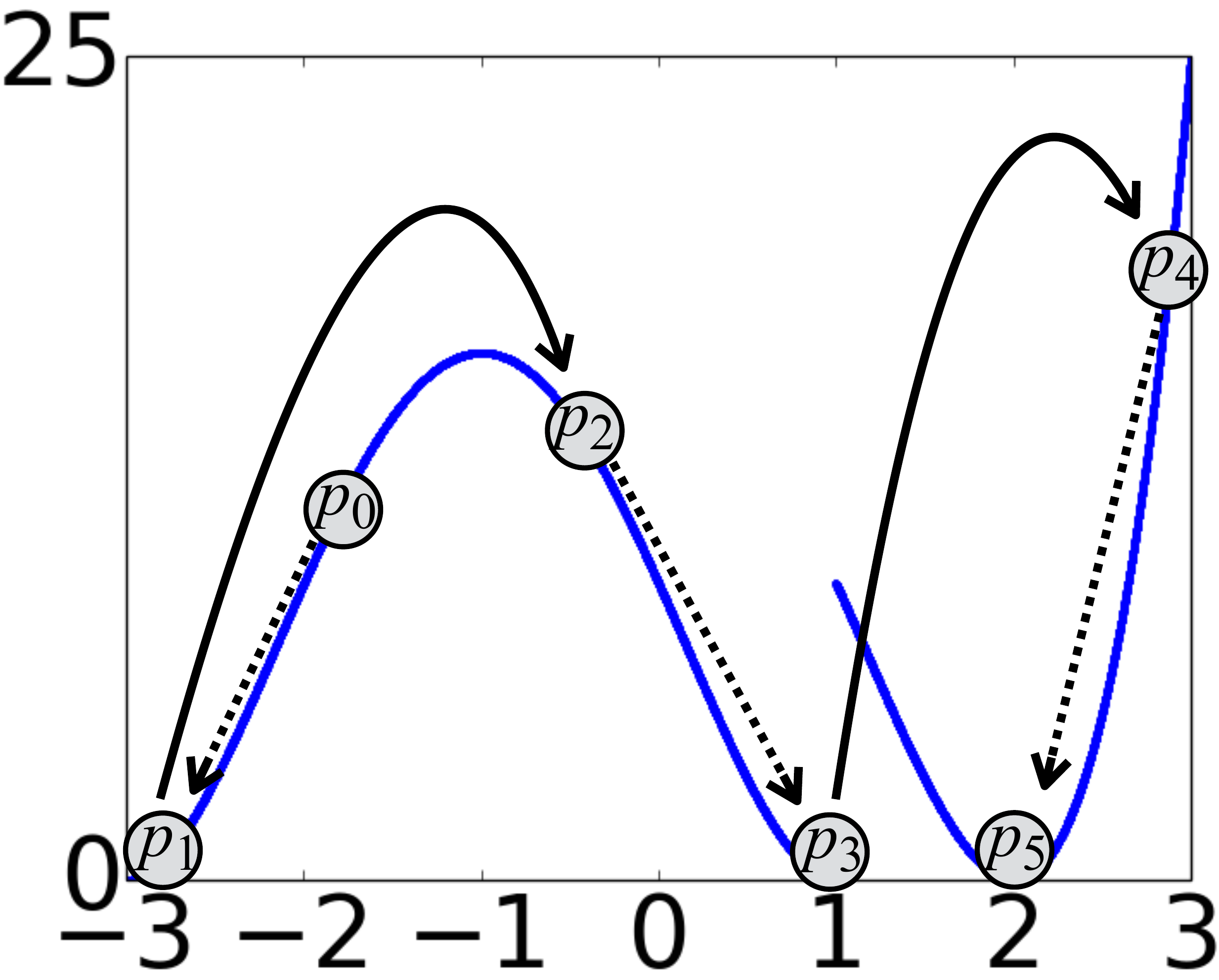}
  \caption{}
\end{subfigure}
\caption{ (a) Local optimization example with objective function
  $\lambda x. x\leq 1~?~0:(x-1)^2$.  The local optimization algorithm
  uses tangents of the curve to quickly converge to a minimum point;
  (b) Global optimization example with objective function
  $\lambda x.x \leq 1~?~((x+1)^2-4)^2:(x^2-4)^2$.  The MCMC method
  starts from $p_0$, converges to local minimum $p_1$, performs a
  Monte-Carlo move to $p_2$ and converges to $p_3$. Then it moves to
  $p_4$ and converges to $p_5$.  }
\label{fig:background:1}
 \end{figure}

 % \begin{remark}
 %   This section discusses unconstrained programming where the unknowns
 %   are reals, as opposed to floating-point values.  In practice, we
 %   solve unconstrained programming of floating-point programs by
 %   regarding them  as real-valued functions and then rounding
 %   the results to floating-point values as necessary.
 % \end{remark}

\section{Overview}
 \label{sect:overview}

\begin{figure}[t]
\centering
\includegraphics[width=1.0\linewidth]{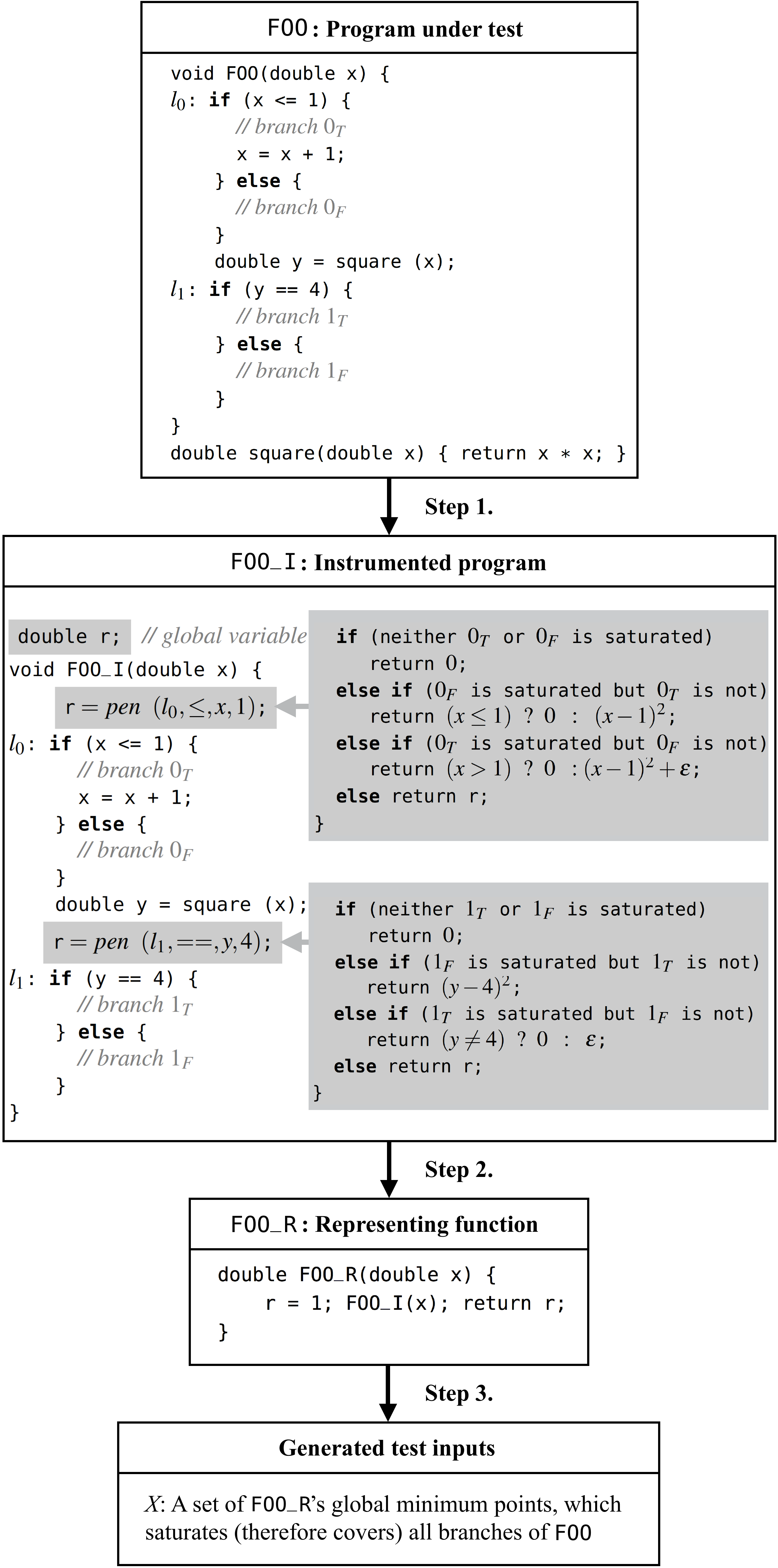}
\caption{An illustration of our approach.  The goal is to find inputs that   saturate (therefore cover) all branches of \FOO, \ie,
  $\{0_T,0_F,1_T,1_F\}$. }
\label{fig:algo:injecting_pen}
\end{figure}

This section states the problem and illustrates our solution. 

\Paragraph{Notation}
  Let \FOO be the program under test with $N$ conditional statements,
  labeled by $l_0$,\ldots,$l_{N-1}$.  Each $l_i$ has a {\tt true}
  branch $i_T$ and a {\tt false} branch $i_F$. We write $\dom{\FOO}$
  to denote the input domain of program $\FOO$.

\subsection{Problem Statement}
\label{sect:overview:pb}

\begin{definition}
  The problem of \emph{branch
  coverage-based testing} aims to find a set of inputs
  $X\subseteq \dom{\FOO}$ that \emph{covers} all branches of
  \FOO.  Here, we say a branch is ``covered'' by $X$ if it is passed
  through by executing \FOO with an input of $X$. 

  We scope the problem with three assumptions. They will be partially
  relaxed in our implementation (Sect. \ref{sect:implem}): 
\begin{itemize}%[nosep]
\item[(a)]  The inputs of  \FOO are floating-point numbers; 
%\item[(b)]  Each Boolean condition in \FOO  is an arithmetic comparison in the form of $a~op~b$, where $\myop\in\{==, \leq, <, \neq, \geq, >\}$, and $a$ and $b$ are floating-point variables or constants.
\item[(b)]  Each Boolean condition in \FOO  is an arithmetic comparison 
between two floating-point variables or constants; and 
\item[(c)]   Each branch of  \FOO is feasible, \ie, it is covered by an input of program \FOO. 
%it can be passed through by executing \FOO with an input $x\in\dom{\FOO}$. 
\end{itemize}
%Assumptions (a) and (b) are set for modeling floating-point code.
%Assumption (c) is set to simplify our presentation. 
\label{def:overview:branchcoverage}
\end{definition}

% %it can be passed through by executing \FOO with an input $x\in\dom{\FOO}$. 
% \end{itemize}

% \begin{itemize}[align=left]
% \item [A1]  The inputs of  \FOO are floating-point numbers (as in our benchmark programs. See Sect.~\ref{sect:eval}.) 
% \item [A2]  Each Boolean condition in \FOO  is an arithmetic comparison. 
% %between two floating-point variables/constants.
% \item [A3]  Each branch of  \FOO is feasible, namely, it can be covered by $\dom{\FOO}$.

% %it can be passed through by executing \FOO with an input $x\in\dom{\FOO}$. 
% \end{itemize}

% \begin{itemize}
% \item[] A1.  The inputs of  \FOO are floating-point numbers (as in our benchmark programs. See Sect.~\ref{sect:eval}.) 
% \item[] A2.  Each Boolean condition in \FOO  is an arithmetic comparison. 
% %between two floating-point variables/constants.
% \item[] A3.   Each branch of  \FOO is feasible, namely, it can be covered by $\dom{\FOO}$.

% %it can be passed through by executing \FOO with an input $x\in\dom{\FOO}$. 
% \end{itemize}

The concept below is crucial. It allows us to rewrite the branch
coverage-based testing problem as defined in
Def.~\ref{def:overview:branchcoverage} to an equivalent, but
easier-to-solve one.

%We introduce the concept of a \emph{saturated branch} and use it to reformulate Def.~\ref{def:overview:branchcoverage}.
\begin{definition}%[Saturated branch]
  Let $X$ be a set of inputs generated during the testing process.
  We say that a branch is \emph{saturated} by $X$ if the branch itself and
  all its descendant branches, if any, are covered by $X$. Here,  a branch
  $b'$ is called a descendant branch of $b$ if there exists a
   control flow  from $b$ to $b'$.  We write
\begin{align}
\Explored(X)
\end{align}
for the set of branches saturated by $X$. 

%WRONG! Essentially, $b\in\Explored(X)$ if and only if   $\forall~b', b\rightarrow^* b'$ implies $b'\in \Explored(X)$. 
\label{def:overview:saturate}
\end{definition}
%In words, a branch is saturated if it is ``deeply covered''. 

%\begin{tabular}[b]{lc}
\noindent \begin{minipage}{0.69\linewidth}
%\mbox{}
\paragraph{}
The control-flow graph on the right illustrates
Def.~\ref{def:overview:saturate}. Suppose that an input set $X$ covers
$\{0_T,0_F,1_F\}$. Then $\Explored(X)= \{0_F,1_F\}$. Branch $1_T$ is
not saturated because it is not covered; branch $0_T$ is not saturated
neither because its descendant branch $1_T$ is not covered.

  \end{minipage}%
%&%
  \begin{minipage}{0.3\linewidth}
\flushright   \includegraphics[width=0.95\linewidth]{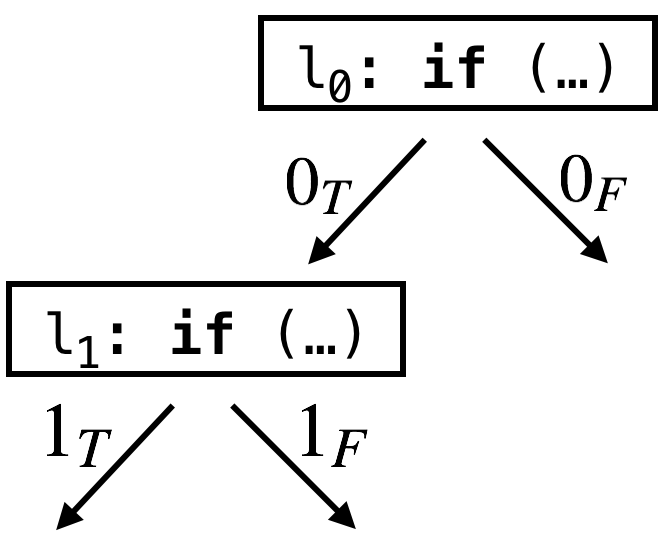}
\end{minipage}
%\end{tabular}

%=============

% To illustrate, suppose an  input set $X$ covers $\{0_T,0_F,1_F\}$ in the
% figure below, then $\Explored(X)= \{0_F,1_F\}$; $1_T$ is not saturated because it is not covered, and $0_T$ is not
% saturated because its child branch $1_T$ is not covered.
% \begin{center}
% \includegraphics[width=0.3\linewidth]{fig/saturated}
% \end{center}

 \paragraph{}

Our approach reformulates the branch coverage-based testing problem
with the lemma below.

\begin{lem} Let $\FOO$ be the program under test. Assume Def.~\ref{def:overview:branchcoverage}(a-c) hold.  Then, a set of
  inputs $X \subseteq \dom{\FOO}$ saturates all $\FOO$’s branches iff
  $X$ covers all $\FOO$'s branches. 

  By consequence, the goal of branch coverage-based testing defined
  in Def.~\ref{def:overview:branchcoverage} can be equivalently stated:
  To find a set of inputs $X\subseteq \dom{\FOO}$ that saturates
  all \FOO's branches.
\label{def:overview:pb}
\end{lem}

%\begin{remark} 

% Following  Lem.~\ref{def:overview:pb},  we reformulate Def. X as below. 
% \begin{definition}
% The goal of branch coverage based testing is to   find a small  set of inputs
%    $X\subseteq \dom{\FOO}$ that saturates all \FOO's branches.
% \label{def:overview:branchcoverage2}
% \end{definition}
% In Def.~\ref{def:overview:branchcoverage2},   we expect the generated input set $X$ to be ``small'', because otherwise we can use $X=\dom{\FOO}$ which already saturates all branches under the assumption Def.~\ref{def:overview:branchcoverage}(c).

 % \emph{the goal of branch coverage based testing}  can be  equivalently stated as to find a small set of inputs
%  $X\subseteq \dom{\FOO}$ that saturates all \FOO's branches.
%\end{remark}

% \begin{itemize}
% \item[] {\bf Problem C.}   Find a small  set of inputs
%   $X\subseteq \dom{\FOO}$ to saturate all \FOO's branches.
% \end{itemize}
%This section proposes a solution to {\bf Problem C}.

%[goal]
\subsection{Illustration of Our Approach}
\label{sect:overview:example}

Fig.~\ref{fig:algo:injecting_pen} illustrates our approach.  The
program under test has two conditional statements $l_0$ and $l_1$.
Our objective is to find an input set that saturates all branches,
namely, $0_T$, $0_F$, $1_T$, and $1_F$.
Our approach proceeds in three steps:

%\paragraph{Step 1. }
%\noindent\emph{Step 1.}
\Paragraph{Step 1}
We inject the global variable $\myr$ in \FOO, and, immediately before
each control point $l_i$, we inject the assignment
\begin{align}
\myr = \pen
\end{align}
where $\pen$ invokes a code segment with parameters associated with
$l_i$.  The idea of $\pen$ is to capture the \emph{distance of the
  program input from saturating a branch that has not yet been
  saturated}.  Observe that in Fig.~\ref{fig:algo:injecting_pen}, both
injected $\pen$ return different values depending on whether the
branches at $l_i$ are saturated or not. \FOOI denotes the instrumented program.

%The key in Step 1 is to
%design $\pen$ to meet certain conditions that allow us to approach the
%problem defined in Lem.~\ref{def:overview:pb} as a unconstrained optimization
%problem. We will specify the conditions  in the next step.

% It has the same input and
% output domains as \FOO. At the end of executing $\FOOI(x)$ for a given
% input $x$, we obtain the value of $\myr$.  Thus, the mapping from $x$
% to $\myr$ forms a mathematical function $r(x)$.

% The key in Step 1 is to design $\pen$ so that the two conditions
% regarding $r(x)$ hold: (1) $r(x)\geq 0$ for all $x$, and (2) $r(x)=0$
% if and only if $x$ \emph{saturates a new branch} in \FOO. In other
% words, a branch that has not been saturated by the generated input set
% $X$ becomes saturated with $X\cup \{x\}$, \ie,
% $\Explored(X)\neq \Explored(X\cup\{x\})$.

%It can be verified that $\FOOI$ shown in the figure satisfies the two conditions. 

%As we will see, the two conditions allow us to solve {\bf Problem C} as  MO problems.

%\paragraph{Step 2.}
%\noindent\emph{Step 2.}
\Paragraph{Step 2} This step constructs the representing function
that we have mentioned in Sect.~\ref{sect:intro}.  The representing
function is the driver program \FOOR shown in
Fig.~\ref{fig:algo:injecting_pen}. It initializes $\myr$ to $1$,
invokes \FOOI and then returns $\myr$ as the output of \FOOR. That is
to say, $\FOOR(x)$ for a given input $x$ calculates the value of
$\myr$ at the end of executing $\FOOI(x)$.

The key in Steps 1 and 2  is to
design $\pen$ so that \FOOR meets the two  conditions below:
\begin{description}[align=left]
\item [C1.] $\FOOR(x)\geq 0$ for all $x$, and 
\item [C2.]   $\FOOR(x)=0$ iff $x$ saturates a new branch.  In other
words, a branch that has not been saturated by the generated input set
$X$ becomes saturated with $X\cup \{x\}$, \ie,
$\Explored(X)\neq \Explored(X\cup\{x\})$.
\end{description}
Conditions C1 and C2 are essential because they allow us to transform a
branch coverage-based testing problem to an unconstrained programming
problem.   Ideally, we can then saturate all branches of \FOO by
repeatedly minimizing \FOOR as shown in the step below.

% In other words, 
% The value of $\myr$ is then retrieved through a driver program \FOOR (Fig.~\ref{fig:algo:injecting_pen}), called the
% representing function  in Sect.~\ref{sect:intro}.  This
% program initializes $\myr$ to $1$, invokes \FOOI and then returns
% $\myr$.  
%  Because its returned value is
% set to be the value of $\myr$ at the end of executing $\FOOI$, the two requirements regarding $r(x)$ in Step 1 can  be expressed as necessary conditions  on \FOOR: 
% \begin{itemize}
% \item[C1.] $\FOOR(x)\geq 0$ for all $x$, and 
% \item[C2.]   $\FOOR(x)=0$ if and only if $x$ saturates a new branch.
% \end{itemize}
% %The two condtions are
% %necessary for the correctness of our approach, which we explain below.

% The program \FOOR plays a crucial role in our approach. It has the
% same input domain as \FOO (as \FOOI) and outputs  a {\tt double}. We
% use \FOOR to transform the branch coverage based testing problem to an MO problem:
% {C2} says $x$ saturates a new branch if and only if $\FOOR$ attains
% $0$ at $x$. Since $\FOOR$ can not be negative by {C1}, the problem of
% saturating a new branch transforms  into  finding a minimum point of
% $\FOOR$. % ( \ie, Eq. \eqref{eq:intro:me} holds. 

%\paragraph{Step 3.}
 %We use MCMC to minimize \FOOR. 

\Paragraph{Step 3}
We calculate the minimum points of \FOOR via unconstrained
programming algorithms described in Sect.~\ref{sect:background}.
Typically, we start with an input set $X=\emptyset$ and
$\Explored(X)=\emptyset$. We minimize \FOOR and obtain a minimum point
$x^*$ which necessarily saturates a new branch by condition C2. Then
we have $X=\{x^*\}$ and we minimize \FOOR again which gives another
input $x^{**}$ and $\{x^*,x^{**}\}$ saturates a branch that is not
saturated by $\{x^*\}$.  We continue this process until all branches
are saturated. Note that when Step 3 terminates, $\FOOR(x)$ must be
strictly positive for any input $x$, due to C1 and C2.

% In this step, we use MCMC to calculate the minimum points of
% \FOOR. Other mathematical optimization techniques, \eg, genetic
% programming~\cite{koza1992genetic}, may also be applicable, which we
% leave for future investigation.  

% \begin{algorithm}\footnotesize
% \caption{Generating test inputs for branch coverage based testing (Sect.~\ref{sect:overview:example}, Step 3)}\label{algo:idealcase} 
% \DontPrintSemicolon
% \KwIn{\FOOR: Representing function of \FOO}
% \KwOut{$X$: A set of generated test inputs}
% %\tcc{$X$: test inputs to generate.}
% Initialize $X=\emptyset$\;
% \Repeat{ \upshape  \FOOR$(x^*)>0$}{
% %Find ipnut $x$ that saturates a new branch $br$ (\ie, $br\not\in\Explored$)\;
% Minimize \FOOR using MCMC. Let $x^*$ be the found minimum point\;
% Update $X$ to $X\cup\{x^*\}$ \;
% }
% \Return $X$\;

% \end{algorithm}

%\paragraph{Put it all together.}

\paragraph{}
Tab.~\ref{tab:algo:example_illustration} illustrates a scenario of how
our approach saturates all branches of program \FOO given in
Fig.~\ref{fig:algo:injecting_pen}. Each ``\textbf{\#n}'' below
corresponds to a line in the table.  We write $\pen_0$ and $\pen_1$
to distinguish the two $\pen$ injected at $l_0$ and $l_1$
respectively.  (\textbf{\#1}) Initially, no branch has been saturated.
Both $\pen_0$ and $\pen_1$ set $\myr=0$, and $\FOOR$ returns $0$ for
any input. Suppose $x^*=0.7$ is found as the minimum point.
(\textbf{\#2})  The branch $1\myF$ is now saturated and $1\myT$ is not. Thus, $\pen_1$ sets $\myr = (y-4)^2$.  Minimizing $\FOOR$ 
gives $x^*=-3.0$, $1.0$, or $2.0$. We have illustrated how these minimum points can be computed in unconstrained programming in Fig.~\ref{fig:background:1}(b).  Suppose $x^*=1.0$ is found.
(\textbf{\#3}) Both $1\myT$ and $1\myF$, as well as $0_T$, are
saturated by the generated inputs $\{0.7,1.0\}$.  Thus, $\pen_1$
returns the previous $\myr$ and $\FOOR$ amounts to $\pen_0$, which
returns $0$ if $x>1$ or $(x-1)^2 + \epsilon$ otherwise, where
$\epsilon$ is a small predefined constant (Sect. \ref{sect:algo}).
Suppose $x^* = 1.1$ is found as the minimum point.
(\textbf{\#4}) All branches have been saturated. In this case, both
$\pen_0$ and $\pen_1$ return $\myr$. $\FOOR$ returns $1$ for all $x$
since $\FOOR$ initializes $\myr$ as $1$.  Suppose the minimum found is
$x^*=-5.2$.  It necessarily satisfies $\FOOR(x^*)>0$, which confirms
that all branches have been saturated (due to C1 and C2).

\begin{table}\footnotesize 
\setlength{\tabcolsep}{1.5pt}%
  \centering
  \caption{
    A scenario of how our approach saturates all branches of \FOO by repeatedly minimizing
    \FOOR.  Column ``$\Explored$'':  Branches that have been saturated. Column ``$\FOOR$'': The representing function and its plot. Column ``$x^*$'': The point where $\FOOR$ attains the minimum.  Column ``$\myX$'': Generated test inputs.}\label{tab:algo:example_illustration}. 
\scriptsize
  \begin{tabular}{l ccccc}
\toprule
 \#~~~&  $\Explored$  &$\FOOR$ & &  $x^*$  & $\myX$ \\ \midrule 
%1
1& $\emptyset$ & $\lambda x. 0$ &\raisebox{-.5\height}{\includegraphics[width=.27\linewidth]{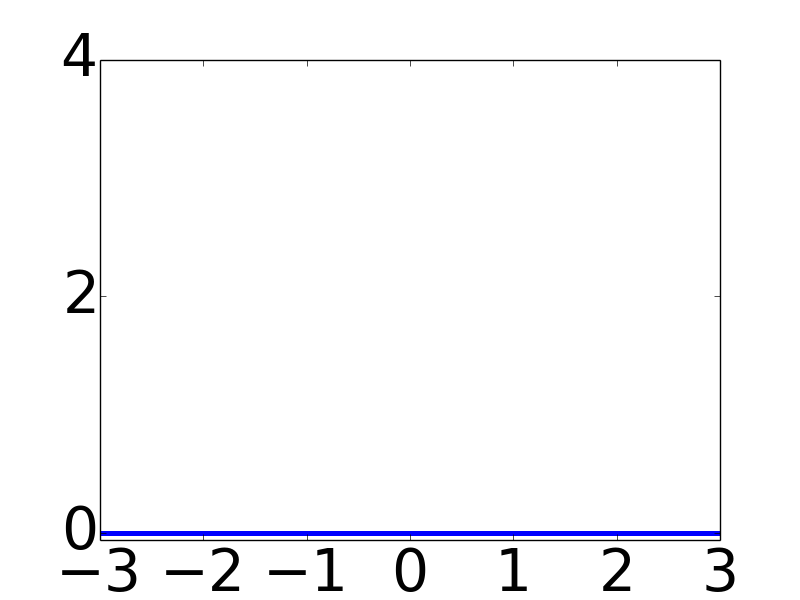}} &$0.7$ & $\{0.7\}$ \\  \arrayrulecolor{light-gray}\hline
%2
2& $\{1\myF\}$ & $\lambda x.\begin{cases}((x+1)^2-4)^2& x\leq 1 \\ (x^2-4)^2 & \text{else} \end{cases}$ &\raisebox{-.5\height}{\includegraphics[width=.27\linewidth]{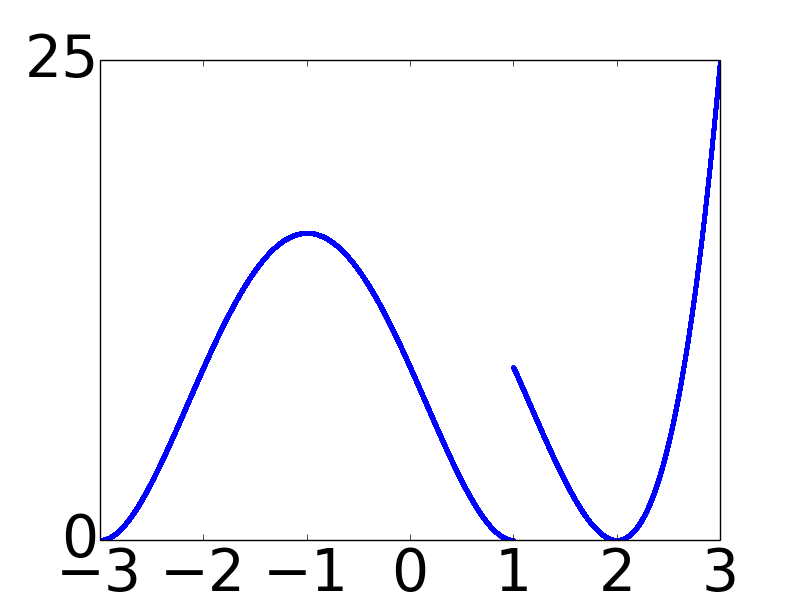}} &$1.0$ & $\{0.7,1.0\}$ \\  \hline
%3
%3& $\{1\myT, 1\myF\}$ & $\lambda x. 0$  &\raisebox{-.5\height}{\includegraphics[width=.27\linewidth]{fig/coverage1}} &$(1.1,0)$ & \parbox{3em}{$\{0.7,1.0,\\ 1.1\}$} \\  \hline
%4
3& \parbox{3em}{$\{0\myT,1\myT,\\ 1\myF\}$} & $\lambda x. \begin{cases}0 & x>1 \\ (x-1)^2+\epsilon &\text{else} \end{cases}$ &\raisebox{-.5\height}{\includegraphics[width=.27\linewidth]{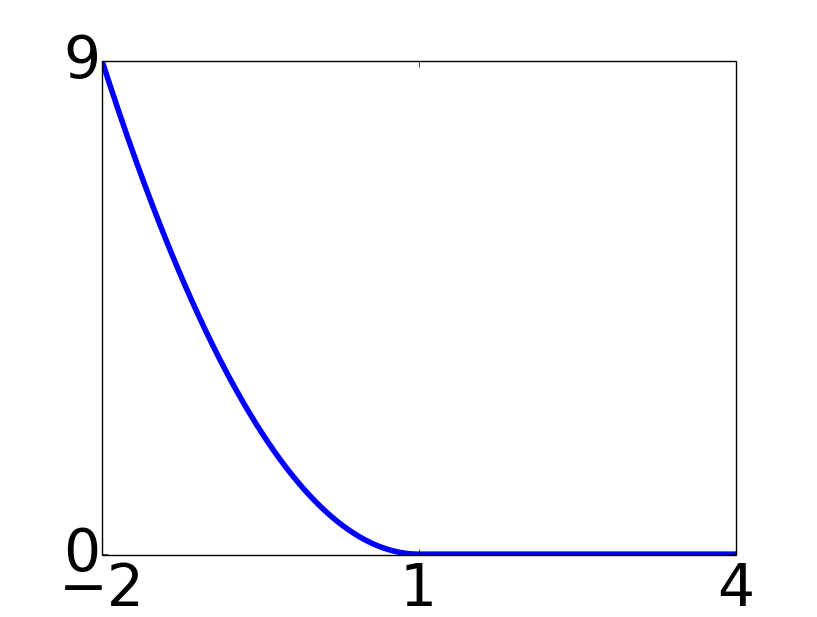}} &$1.1$ & \parbox{3em}{$\{0.7,1.0,$ \\ $1.1\}$} \\  \hline
%5
4&\parbox{3em}{$\{0\myT,1\myT,\\ 0\myF, 1\myF\}$} & $\lambda x. 1$ &\raisebox{-.5\height}{\includegraphics[width=.27\linewidth]{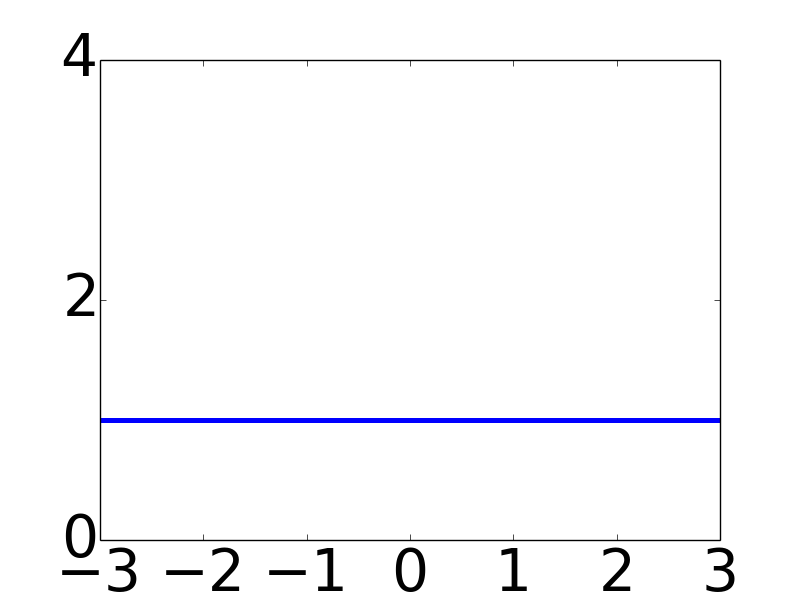}} &$-5.2$ & \parbox{3em}{$\{0.7,1.0,\\1.1,-5.2\}$} \\  %\hline
\arrayrulecolor{black}
\bottomrule
  \end{tabular}

\end{table}

% \begin{remark}

% Given an input $x$, the value of $\FOOR(x)$ may change during the
% minimization process.  In fact, \FOOR is constructed with
% injected $\pen$ which returns different values at $l_i$ depending on whether the
% branches $i_T$ and $i_F$ have been saturated. Thus, the minimization step in
% our algorithm differs from existing  unconstrained programming techniques where the objective function is
% fixed~\cite{Zoutendijk76}.

% \end{remark}

%%%%%%%%%%%% to put to sect. theory
% In theory, because each minimal point of \FOOR saturates a new branch
% until all branches are saturated, Algo.~\ref{algo:idealcase}
% necessarily terminates with input set $X$ that saturates all
% branches. In practice, however, global MO is intractable, and none of the
% existing tools can guarantee that Line 3 of Algo.~\ref{algo:idealcase}
% calculates  the minimum point accurately. In the case where $x^*$
% calculated on Line~3 fails to be accurate, Algo.~\ref{algo:idealcase}
% may need to be interrupted or time out, yielding partial coverage of
% \FOO's branches. We will formalize and discuss this phenomenon of
% \emph{incompleteness} in Sect.\ref{sect:theory} and
% Sect.~\ref{sect:discussion}.

%\paragraph{Remark 2.}

 \section{Algorithm}
 \label{sect:algo}
 
We provide details corresponding to the three steps in
Sect.~\ref{sect:overview:example}. The algorithm is summarized in
Algo.~\ref{theory:algo:coverme}.

\Paragraph{Algorithm for Step 1} 
The outcome of this step is the instrumented program \FOOI.  As
explained in Sect.~\ref{sect:overview:example}, the essence is to
inject the variable $\myr$ and the assignment $\myr=\pen$ before each conditional
statement (Algo.~\ref{theory:algo:coverme}, Lines 1-4).

%  We show the high-level steps in
% Algo.~\ref{theory:algo:coverme}, Lines 1-4, and the design of \pen in
% Algo.~\ref{theory:algo:coverme}, Lines 14-23.

% This step constructs \FOOI (Algo.~\ref{theory:algo:coverme}, Lines 1-4).
% As explained in Sect.~\ref{sect:overview:example}, we introduce
% a global variable $\myr$ and inject its assignments before conditional
% statements. We adopt a ``lazy'' injection by invoking 
% $\pen$ before all conditional statements. 

To define $\pen$, we first introduce a set of helper functions that
are sometimes known as \emph{branch distance}. There are many
different forms of branch distance in the
literature~\cite{Korel:1990:AST:101747.101755,McMinn:2004:SST:1077276.1077279}.
We define ours with respect to an arithmetic condition $a~op~b$.

% a parameter $\epsilon\in\Real_{> 0}$.
\begin{definition}\label{def:overview:bd}
Let $a,b\in\Real$, $\myop\in\{==,\leq,<,\neq,\geq,>\}$, $\epsilon\in\Real_{> 0}$.
 We define \emph{branch distance} $d_{\epsilon}(\myop, a,b)$ as follows:
\begin{align}
d_{\epsilon}(==, a,b)&\mydef (a-b)^2 \\
d_{\epsilon}(\leq, a,b)&\mydef (a\leq b)~?~0:(a-b)^2\\
d_{\epsilon}(<,a,b)&\mydef (a <b)~?~0:(a-b)^2+\epsilon\\
%d(\geq,a,b)&\mydef a\geq b\ ?\ ~0\ :\ (a - b)^2\\
%d(>, a,b)&\mydef a > b\ ?\ ~0\ :\ (a - b)^2  +\epsilon\\
d_{\epsilon}(\neq, a,b)&\mydef (a\neq b)~?~0:\epsilon 
\end{align}
and $d_{\epsilon}(\geq,a,b)\mydef d_{\epsilon}(\leq,b,a)$,
$d_{\epsilon}(>,a,b)\mydef d_{\epsilon}(<,b,a)$. 
%where $op$, $a$ and
%$b$ refer to the arithmetic comparison operator, the left-hand and
%right-hand operands of the predicate of the conditional statement,
%respectively.  
We use
$\epsilon$ to denote a small positive floating-point close to machine
epsilon. The idea is to treat a strict floating-point inequality $x >
y$ as a non-strict inequality $x \geq y+\epsilon$, \etc
We will drop the explicit reference to $\epsilon$ when
using the branch distance.
\end{definition}

The intention of $d(\myop,a,b)$ is to {quantify how far $a$ and $b$
  are from attaining $a~\myop~b$}. For example, $d(==,a,b)$ is
strictly positive when $a\neq b$, becomes smaller when $a$ and $b$
go closer, and vanishes when $a==b$. The following property holds:
\begin{align}
d(\myop,a,b)\geq 0\;\; \text{~and~}\;\; d(\myop,a,b)=0\Leftrightarrow a~\myop~b.
\label{eq:overview:brDist}
\end{align}

 As an analogue, we set $\pen$
to {quantify how far an input is from saturating a new branch}.
We define $\pen$ following Algo.~\ref{theory:algo:coverme}, Lines~14-23.
\begin{definition}
For  branch coverage-based testing, the function $\pen$ has four parameters, namely, the label of the
conditional statement $l_i$, $\myop$, and $a$ and $b$ from the arithmetic
condition $a~\myop~b$.
\begin{itemize}
\item [(a)]  If neither of the two branches at $l_i$ is
saturated, we let $\pen$ return $0$ because any input saturates a
new branch (Lines~16-17). 
\item [(b)] If one branch at $l_i$ is saturated but the other is
not, we set $\myr$ to be the distance to the unsaturated branch (Lines~18-21).
\item [(c)] If both branches at $l_i$ have already been saturated,
$\pen$ returns the previous value of the global variable $\myr$ (Lines~22-23).
\end{itemize}

For example, $\pen$ at $l_0$ and $l_1$ are invoked as
$pen(l_i, \leq,x,1)$ and $\pen(l_1,==,y,4)$ respectively in
Fig.~\ref{fig:algo:injecting_pen}.

\label{def:overview:pen}
\end{definition}

% (1) If both branches at $l_i$ are
% unsaturated, we let $\pen_i$ return $0$ because any input saturates a
% new branch;  (2) If one branch at $l_i$ is saturated but the other is
% not, we set $\myr$ to be the distance to the branch that is not
% saturated; (3) If both branches at $l_i$ have already been saturated,
% $\pen_i$ returns the previous value of $\myr$.  

\SetKwProg{myproc}{Procedure}{}{} \SetKwInput{Gb}{Global}
\begin{algorithm} [t]\footnotesize 
  \DontPrintSemicolon
  \KwIn{\begin{tabularx}{.8\linewidth}[t]{>{$}l<{$} X}
      \FOO & Program under test \\
      \nStart & Number of starting points\\
      \LM & Local optimization  used in \mcmc\\
      \niter & Number of iterations for \mcmc\\
    \end{tabularx}} 

\KwOut{\begin{tabularx}{.8\linewidth}[t]{>{$}l<{$} X}
      X & Generated input set \\
    \end{tabularx}}

\BlankLine
  \tcc{Step 1}
Inject   global variable $\myr$  in $\FOO$\; 
\For{ conditional
    statement $l_i$ in {\upshape $\FOO$} }{ 
    Let the  Boolean condition at  $l_i$ be 
    $\mylhs~\myop~\myrhs$ where $\myop\in\{\leq,<,=,>,\geq,\neq\}$\;
    Insert assignment $\myr=\pen(l_i,\myop,\mylhs, \myrhs)$ before $l_i$ }
  \tcc{Step 2}
  Let $\FOOI$ be the newly instrumented program, and $\FOOR$ be:
  \lstinline!  double FOO_R(double x) {r = 1; FOO_I(x); return r;}!\;

\tcc{Step 3}
Let $\Explored=\emptyset$\;
Let $X =\emptyset$ \; 
\For{$k=1$ to $\nStart$}{
Randomly take a  starting point $x$\;
%\tcc{Global minimization with start point $x$}
Let $x^* = \text{\mcmc}(\FOOR,x)$\;
  \lIf{\upshape $\FOOR(x^*) = 0$}{ $ X = X \cup \{x^*\}$  }
Update $\Explored$ \;
}

\Return $X$\;

  \BlankLine

\SetKwProg{Fn}{Function}{}{}

\Fn{$\pen(l_i,\myop, \mylhs,\myrhs)$}{
Let $i\myT$ and $i\myF$ be the {\tt true} and the {\tt false} branches at $l_i$\;
\If{$i\myT\not\in\Explored$ and $i\myF\not\in\Explored$ }
{\Return $0$}
\uElseIf{$i\myT\not\in\Explored$ and $i\myF\in\Explored$ }
{\Return $\dist(\myop,\mylhs,\myrhs)$ 
 \tcc{$d$: Branch distance}
}
\uElseIf{$i\myT\in\Explored$ and $i\myF\not\in\Explored$ }
{\Return $\dist(\overline{\myop}, \mylhs,\myrhs)$  \tcc{$\overline{\myop}$: the opposite of $\myop$}}
\uElse(\tcc*[h]{$i\myT\in\Explored$ and $i\myF\in\Explored$ }){\Return $\myr$ }
}

  \BlankLine
 \Fn{{\upshape \mcmc}($\energy$, $x$)}
 {

     $\xloc=\lmin(\energy, x)$\; \tcc{Local minimization}  

     \For{$k = 1$ \KwTo $\niter$}{

       Let $\delta$ be a random perturbation generation from a predefined distribution\;
       Let  $\xpro=\lmin(\energy, \xloc + \delta)$\;

       \lIf{$\energy(\xpro)<\energy(\xloc)$}{$\mathit{accept} = \mathit{true}$}

       \Else{

         Let $m$ be a random number  generated from the uniform distribution on $[0,1]$\;

%Let  be the Boolean
         Let ${\mathit{accept}}$ be the Boolean $m<\exp(\energy(\xloc) - \energy(\xpro)) $ 
        }

       \lIf{$\mathit{accept}$}{ $\xloc = \xpro$}
     }    

     \Return $\xloc$\;
 }

  \caption{Branch coverage-based testing}
  \label{theory:algo:coverme} 
\end{algorithm}

\Paragraph{Algorithm for Step 2}
This step constructs the representing function \FOOR
(Algo.~\ref{theory:algo:coverme}, Line~5). Its input domain is the
same as that of \FOOI and \FOO, and its output domain is {\tt double},
so to simulate a real-valued mathematical function which can then be
processed by the unconstrained programming  backend.

\FOOR initializes $\myr$ to 1. This is essential for the correctness of the
algorithm because we expect \FOOR returns a non-negative value when all branches
are saturated (Sect.~\ref{sect:overview:example}, Step 2).
\FOOR then calls $\FOOI(x)$ and records the value of $\myr$ at the end
of executing $\FOOI(x)$. This $\myr$ is the returned value of \FOOR.

As mentioned in Sect.~\ref{sect:overview:example}, it is important to
ensure that \FOOR meets conditions C1 and C2.  The condition {C1}
holds true since \FOOR returns the value of the instrumented $\myr$, which
is never assigned a negative quantity. The theorem below states \FOOR
also satisfies C2.

\begin{thm}
Let \FOOR be the program constructed in Algo.~\ref{theory:algo:coverme},  and $S$ the branches that have been saturated. Then, for any input $x\in\dom{\FOO}$, $\FOOR(x)=0$ $\Leftrightarrow $  $x$ saturates a branch that does not belong to $S$. 
\label{lem:overview:c2}
\end{thm}
\begin{proof}
  We first prove the $\Rightarrow$ direction.  Take an arbitrary $x$
  such that $\FOOR(x)=0$. Let $\tau=[l_0,\ldots l_n]$ be the path in
  $\FOO$ passed through by executing $\FOO(x)$. We know, from Lines
  2-4 of the algorithm, that each $l_i$ is preceded by an invocation
  of $\pen$ in $\FOOR$. We write $\pen_i$ for the one injected before
  $l_i$ and divide $\{\pen_i\mid i\in[1,n]\}$ into three groups. For
  the given input $x$, we let {\it P1}, {\it P2} and {\it P3} denote the groups of $\pen_i$ that
  are defined in Def.~\ref{def:overview:pen}(a),
  (b) and (c),
  respectively.  Then, we can always have a prefix path of
  $\tau=[l_0,\ldots l_m]$, with $0\leq m\leq n$ such that each
  $\pen_i$ for $i\in[m+1,n]$ belongs to {\it P3}, and each $\pen_i$ for
  $i\in [0,m]$ belongs to either {\it P1} or {\it P2}.  Here, we can guarantee the
  existence of such an $m$ because, otherwise, all $\pen_i$ belong in
  {\it P3}, and $\FOOR$ becomes $\lambda x.1$. The latter contradicts the
  assumption that $\FOOR(x)=0$.
  Because each $\pen_i$ for $i>m$ does nothing but performs
  $\myr=\myr$, we know that $\FOOR(x)$ equals to the exact value of
  $\myr$ that $\pen_m$ assigns. Now consider two disjunctive cases on
  $\pen_m$. If $\pen_m$ is in {\it P1}, we immediately conclude that
  $x$ saturates a  new branch. Otherwise, if $\pen_m$ is in {\it P2}, we
  obtains the same from Eq.~\eqref{eq:overview:brDist}.
Thus,  we have established the $\Rightarrow$
  direction of the theorem.

  To prove the $\Leftarrow$ direction, we use the same
  notation as above, and let $x$ be the input that saturates a new branch,
  and $[l_0,\ldots, l_n]$ be the exercised path. Assume that $l_m$
  where $0\leq m\leq n$ corresponds to the newly saturated branch. We
  know from the algorithm  that (1) $\pen_m$ updates
  $\myr$ to $0$, and (2) each $\pen_i$ such that $i>m$ maintains the
  value of $\myr$ because their descendant branches have been saturated.
We have thus proven the $\Leftarrow$ direction of the theorem.
\end{proof}

% -----------------
% Because each minimal point of \FOOR saturates a new branch according to Lem. X, 

% until all branches are saturated, Algo.~\ref{algo:idealcase}
% necessarily terminates with input set $X$ that saturates all
% branches. In practice, however, global MO is intractable, and none of the
% existing tools can guarantee that Line 3 of Algo.~\ref{algo:idealcase}
% calculates  the minimum point accurately. In the case where $x^*$
% calculated on Line~3 fails to be accurate, Algo.~\ref{algo:idealcase}
% may need to be interrupted or time out, yielding partial coverage of
% \FOO's branches. We will formalize and discuss this phenomenon of
% \emph{incompleteness} in Sect.\ref{sect:theory} and
% Sect.~\ref{sect:discussion}.

\Paragraph{Algorithm for Step 3}
The main loop (Algo.~\ref{theory:algo:coverme}, Lines 8-12) relies on
an existing MCMC engine. It takes an objective function and a starting
point and outputs $x^*$ that it regards as a minimum point.  Each
iteration of the loop launches MCMC from a randomly selected starting
point (Line 9). From each starting point, MCMC computes the minimum
point $x^*$ (Line 10). If $\FOOR(x^*)=0$, $x^*$ is to be added to  $X$ (Line 11).  Thm.~\ref{lem:overview:c2}
ensures that $x^*$ saturates a new branch if
$\FOOR(x^*)=0$. Therefore, in theory, we only need to set
$\nStart=2*N$ where $N$ denotes the number of conditional statements,
so to saturate all $2*N$ branches.  In practice, however, we set
$\nStart>2*N$ because MCMC cannot guarantee that its output is a true
global minimum point.

%~\footnote{Our implementation generates different random
%  starting points from the Hypothesis library~\cite{hypothesis:web}.
%  It not only produces floating-point numbers as standard
%  pseudo-random number generator, but it also produces those near the
%  boundary of the input range, such as {\num{1E-308}}.}  

The MCMC procedure (Algo.~\ref{theory:algo:coverme}, Lines 24-34) is
also known as the Basinhopping
  algorithm~\cite{leitner1997global}. It is an MCMC sampling over the
space of the local minimum points~\cite{Li01101987}.  The random
starting point $x$ is first updated to a local minimum point $\xloc$
(Line 25). Each iteration (Lines 26-33) is composed of the two phases
that are classic in the {Metropolis-Hastings} algorithm family of
MCMC~\cite{citeulike:831786}. In the first phase (Lines 27-28), the
algorithm \emph{proposes} a sample $\xpro$ from the current sample
$x$. The sample $\xpro$ is obtained with a perturbation $\delta$
followed by a local minimization, \ie, $\xpro=\LM(f,\xloc+\delta)$
(Line 28), where \LM denotes a local minimization in Basinhopping, and
$f$ is the objective function.  The second phase (Lines 29-33) decides
whether the proposed $\xpro$ should be accepted as the next sampling
point. If $\energy(\xpro)<\energy(\xloc)$, the proposed $\xpro$ will
be sampled; otherwise, $\xpro$ may still be sampled, but only with the
probability of $\exp((\energy(\xloc)-\energy(\xpro))/T)$, in which $T$
(called the annealing temperature~\cite{Kirkpatrick83optimizationby})
is set to $1$ in Algo.~\ref{theory:algo:coverme} for simplicity.

\section{Implementation}
\label{sect:implem}
% =============== to comment out this part
% \clearpage

% \noindent $\FOO$: Program under test \\
% in any LLVM-supported language
% \begin{lstlisting}
% type_t FOO (double x1, double x2, ...)
% \end{lstlisting}

% \noindent $\pen$  (.cpp)
% \begin{lstlisting}
% double pen (int i, int op, double lhs, double rhs)
% \end{lstlisting}

% \noindent $\FOOI$: Instrumented program (.bc)
% \begin{lstlisting}
% type_t FOO_I (double x1, double x2, ...)
% \end{lstlisting}

% \noindent $\LOADER$ (.cpp)
% \begin{lstlisting}
%  void loader (double* P)
% \end{lstlisting}
% % \begin{lstlisting}
% % void LOADER (double* P){
% %    double x = P[0];
% %    double y = P[1];
% %    FOO_I (x,y);
% % }
% % \end{lstlisting}

% \noindent $\FOOR$ (.cpp)
% \begin{lstlisting}
% void FOO_R (double* P)
% \end{lstlisting}
% % \begin{lstlisting}
% % void FOO_R (double* P){
% %    r = 1;
% %    LOADER (P);
% %    return r;
% % }
% % \end{lstlisting}

% \noindent \librso 
% \begin{lstlisting}
% void FOO_R (double* P)
% \end{lstlisting}

% \noindent MCMC minimization procedure (.py)
% \begin{lstlisting}
% basinhopping (func, sp, n_iter, callback)
% \end{lstlisting}
% % \begin{lstlisting}
% % gen_inputs(n_start)
% % \end{lstlisting}

% LLVM pass

% Linking

\begin{figure}
\centering
\includegraphics[width=1.0\linewidth]{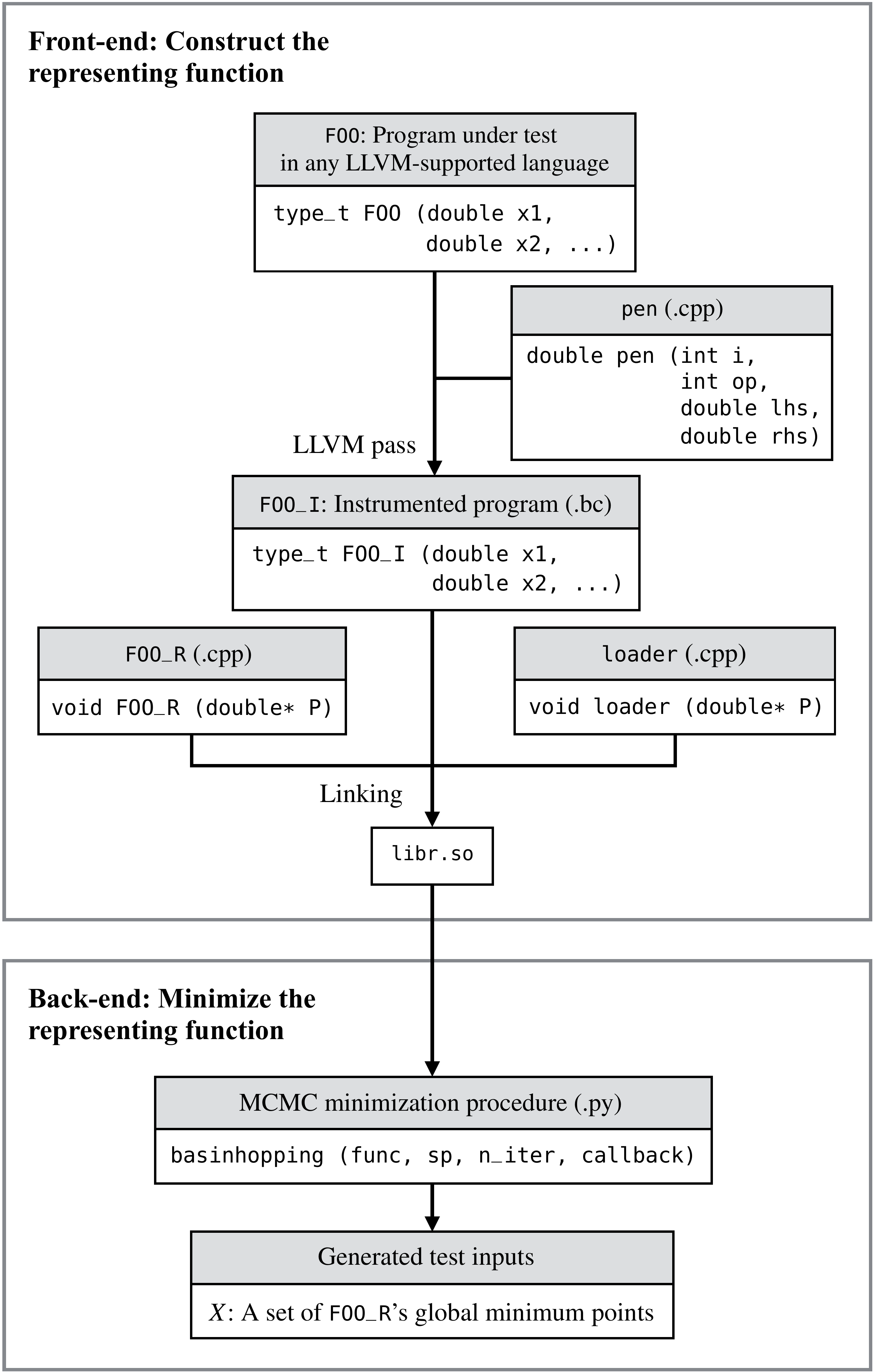}
\caption{\coverme implementation.}
\label{fig:implem:overview}
\end{figure}

As a proof-of-concept demonstration, we have implemented
Algo.~\ref{theory:algo:coverme} into the tool \coverme.  This section
presents its implementation and technical details.

% At the highest-level, \coverme is
%implemented as a front-end part and a back-end part. The front-end 
% constructs the representing function (which corresponds to
%Step 1 and 2 of our proposed solution), and the back-end 
%minimizes it (corresponding to Step 3 of
%our proposed solution).

% Fig. X illusrates major components of our implementation. 
%  Recall that \FOO,
% \FOOI, \FOOR are the program under test, the program instrumented with 
% $\pen$, and the driver program
% that implements the representing functions, respectively. 

% The component \loader serves as an 

% The component \loader has
% not been mentioned. In previous section, where we use one-dimensional
% program inputs to simplify our presentation but \coverme also needs to
% handle multiple dimensional inputs and types {\tt double*} input
% types.  In addition, as we will see, Because the back-end, the MCMC
% engine, takes an objective function of type $\FP^n\rightarrow \FP$, it
% is imperative that the front-end delivers the representing function
% \FOOR following that signature. The compnent \loader is implemented to
% handle these concerns. In words, it ``load'' whatever inputs, which
% can be one single {\tt double} or multiple pointers to {\tt double}s
% to a uniformed {\tt double*}. We provide more details on this
% component shortly.

\subsection{Frontend of CoverMe}
The frontend implements Steps 1 and 2 of Algo.~\ref{theory:algo:coverme}. 
\coverme compiles the program under test \FOO to LLVM IR with Clang~\cite{clang:web}. Then it uses an LLVM pass~\cite{llvmpass:web} to inject assignments. The program under test can be in any LLVM-supported language, \eg, Ada, the C/C++ language family, or Julia, \etc Our implementation has been tested on  C code. 

% \coverme compiles the program under test \FOO to LLVM IR with
% Clang~\cite{clang:web}. Then it use an \emph{LLVM
%   pass}~\cite{llvmpass:web} to inject assignments in \FOO. The program
% under test can be in any LLVM-supported language, \eg, the C/C++
% language family, or Julia. Our current implementation has been
% thoroughly tested on C code.

Fig.~\ref{fig:implem:overview} illustrates \FOO as a function of signature \texttt{type\_t
  FOO (type\_t1 x1, type\_t2 x2, \ldots)}.  The return type (output) of the function,
\texttt{type\_t}, can be any kind of types supported by C, whereas the types of the input
parameters, {\texttt{type\_t1, type\_t2, \ldots}}, are restricted to
{\tt double} or {\tt
  double*}. %(Fig. X illustrates {\tt type\_t1} and {\tt type\_t2} as {\tt double}s, and We will explain how we handle {\tt double*} in the next subsection).
We have explained the signature of $\pen$ in Def.~\ref{def:overview:pen}.
Note that \coverme does not inject $\pen$ itself into \FOO, but instead injects
assignments that invoke $\pen$.  We implement $\pen$ in a separate C++
file.
%Besides, the types of
%these {\tt lhs}, and {\tt rhs} are {\tt double}s, as we have assumed
%in A1. But in our tested program some comparisons are between integers
%or unsigned integers. We will explain in Sect. X below how we handle
%those situations.

% The program under test \FOO contains an entry function to test, of
% signature \texttt{type\_t FOO(type\_t1 x1, type\_t2 x2, \ldots)} where
% the output type, \texttt{type\_t}, can be any, whereas the types for
% the input parameters, {\texttt{type\_t1, type\_t2, \ldots}}, are
% restricted to {\tt double} or {\tt
%   double*}. %(Fig. X illustrates {\tt type\_t1} and {\tt type\_t2} as {\tt double}s, and We will explain how we handle {\tt double*} in the next subsection).
% The signature of $\pen$ has been explained earlier (Sect.~\ref{sect:algo}).  Note
% $\pen$ itself will not be inserted to \FOO, but it will be invoked by
% the assignments $\myr=\pen$ to inject. We implement $\pen$ in a
% separate C++ file.
% %Besides, the types of
% %these {\tt lhs}, and {\tt rhs} are {\tt double}s, as we have assumed
% %in A1. But in our tested program some comparisons are between integers
% %or unsigned integers. We will explain in Sect. X below how we handle
% %those situations.

%As mentioned above, we denote the instrumented program by \FOOI. It has the
%same signature as \FOO. Thus, any valid input of \FOO is also valid for
%\FOOI.  Unlike \FOO,  \FOOI is an LLVM bitcode (.bc file).

The frontend also links \FOOI and \FOOR with  a simple program  \loader into  a shared object file \librso, which is the 
outcome of the frontend. It stores the
representing function in the form of a shared object file (.so file). 

%Besides
%\FOOR, \librso also includes the information regarding \FOO's input
%parameters such as their types and dimensions, which are neceassary to
%the back-end.
 
\subsection{Backend of CoverMe}

The backend implements Step 3 of Algo.~\ref{theory:algo:coverme}. 
It invokes the representing function via \librso.  
The kernel of the backend is an external MCMC engine.
It uses the off-the-shelf implementation known as \emph{Basinhopping} from the Scipy Optimization package~\cite{scipy:web}. Basinhopping takes a range of input parameters. Fig.~\ref{fig:implem:overview} shows the important ones for our implementation {\tt
  basinhopping(f, sp, n\_iter, call\_back)}, where {\tt f} refers to the
representing function  from \librso, {\tt sp} is a starting point as a 
Python {\tt Numpy} array, {\tt n\_iter} is the iteration number used in Algo.~\ref{theory:algo:coverme} and {\tt call\_back} is a
client-defined procedure. Basinhopping invokes {\tt call\_back}  at the end of each
iteration (Algo.~\ref{theory:algo:coverme}, Lines~24-34). The call\_back procedure
allows \coverme to terminate if it saturates all branches. In this
way, \coverme does not need to wait until passing all $\nStart$
iterations (Algo.~\ref{theory:algo:coverme}, Lines~8-12).

%, which we use for terminating our algorithm when all branches have been saturated, and also for dealing with infeasible branches (see below, sect. X); $\niter$ is the number of interations to use in the MCMC iterations. 

%As explained earlier, the MCMC procedure, basinhoping, is to be invoked n\_start times. These itertionas are implemented in gen\_inuts, which takes n\_start as a parameter, and runs basinhopping n\_start times and outputs the calculated inputs. gen\_input is an implementation of line X-X of Algo. x

\subsection{Technical Details}
Sect.~\ref{sect:overview} assumes Def.~\ref{def:overview:branchcoverage}(a-c) for the sake of simplification.   This section discusses how   \coverme relaxes the assumptions when  handling real-world floating-point code.  We also show how \coverme handles function calls at the end of this section. 

\Paragraph{Handling Pointers  (Relaxing Def.~\ref{def:overview:branchcoverage}(a))} 
We consider only pointers to floating-point numbers. They may occur  (1) in an input parameter, (2) in a conditional statement, or  (3) in the code body but not in the conditional statement. 

\coverme inherently handles case (3) because it is execution-based and does not need to analyze pointers and their effects.  \coverme currently does not handle case (2) and   ignores these conditional statements by not injecting $\pen$ before them.  %Note:  The condition such as {\tt if (*x <=3)}  belongs to category (3) because it is the LLVM IR of the program under test uses a temporary variable for {\tt *x} which appears in the code body but not in the conditional statement. 

Below we explain how \coverme deals with case (1). A branch
coverage testing problem for a program whose inputs are pointers to
doubles, can be regarded as the same problem with a simplified program
under test. For instance, finding test inputs to cover branches of
program {\tt void FOO(double* p) \{if (*p <= 1)... \} } can be reduced
to testing the program {\tt void FOO\_with\_no\_pointer (double x)
  \{if (x <= 1)... \}}. \coverme transforms program \FOO to {\tt FOO\_with\_no\_pointer}
if a   \FOO's input parameter is a floating-point pointer.   
% \coverme realizes this problem reduction using
%the \loader program: {\tt void loader(double* P) \{ double y[1]; y[0]
%  = P[0]; FOO\_I(y);\}} where, as before, \FOOI refers the instrumented
%\FOO.

% \coverme naturally handles (2) because 
% conditions such as {\tt if (*x <= 3)} are transformed to {\tt r=*x; if (t<=3)} where  {\tt t} is an tempory variable. \coverme  does not handle category (3), such
% as {\tt if (p != Null)}, so there will be no injected assignment before them.  

% Below, We briefly show how \coverme deals with case (1).  Consider program \FOO
% of interface void {\tt void FOO (double *y)}.  Observe
% that any path of \FOO triggered by  input {\tt y} of type {\tt double*}  can also be triggered by a pointer {\tt P} by  the loader program
% \loader:

% %\coverme systematically uses  to cover  branches of program FOO.

\Paragraph{Handling  Comparison between Non-floating-point Expressions (Relaxing Def.~\ref{def:overview:branchcoverage}(b))} 

We have encountered situations where a  conditional statement invokes a comparison between non floating-point numbers.  \coverme handles these situations by first promoting the non floating-point numbers to floating-point numbers and then injecting $\pen$ as described in Algo.~\ref{theory:algo:coverme}.  For example, before a conditional statement like {\tt if (xi op yi)} where {\tt xi} and {\tt yi} are integers, \coverme injects  {\tt r = pen (i, op, (double) x, (double) y));}. 
Note that such an approach does not allow us to handle  data types that are incompatible with floating-point types, \eg, conditions like {\tt if (p != Null)}, which \coverme has to ignore.

\Paragraph{Handling Infeasible Branches (Relaxing Def.~\ref{def:overview:branchcoverage}(c))}
Infeasible branches are branches that cannot be covered by any input.
Attempts to cover infeasible branches are useless and time-consuming. 

Detecting infeasible branches is a difficult problem in general.
\coverme uses a heuristic to detect infeasible branches.  When CoverMe
finds a minimum that is not zero, it deems the unvisited branch of the
last conditional to be infeasible and adds it to $\Explored$, the set
of unvisited and deemed-to-be infeasible branches.

 Imagine that we modify $l_1$ of the program \FOO in
Fig.~\ref{fig:algo:injecting_pen} to  the conditional statement {\tt
  if (y == -1)}. Then the branch $1_T$ becomes infeasible. We rewrite this
modified program below and  illustrate how we deal with infeasible branches. 
\begin{lstlisting}
l0: if (x <= 1) {x++}; 
    y = square(x); 
l1: if (y == -1) {...}
\end{lstlisting}
where we omit the concrete implementation of {\tt square}. 

Let \FOOR denote the representing function constructed for the
program. In the minimization process, whenever \coverme obtains $x^*$
such that $\FOOR(x^*) > 0$, \coverme selects a branch that it regards
infeasible. \coverme selects the branch as follows: Suppose $x^*$
exercises a path $\tau$ whose last conditional statement is denoted by
$l_z$, and, without loss of generality, suppose $z_T$ is passed
through by $\tau$, then \coverme regards $z_F$ as an infeasible
branch.

In the modified program above, if $1_F$ has been saturated, the
representing function evaluates to $(y+1)^2$ or $(y+1)^2+1$, where $y$
equals to the non-negative {\tt square(x)}. Thus, the minimum point
$x^*$ must satisfy $\FOOR(x^*)>0$ and its triggered path ends with
branch $1_F$.  \coverme then regards $1_T$ as an infeasible branch.

\coverme then regards the infeasible branches as already saturated. It
means, in line 12 of Algo.~\ref{theory:algo:coverme}, \coverme updates
$\Explored$ with saturated branches and infeasible branches (more
precisely, branches that \coverme regards infeasible).

The presented heuristic works well in practice (See
Sect.~\ref{sect:eval}), but we do not claim that our heuristic always
correctly detects infeasible branches.

%We believe that the problem of detecting infeasible branches deserves further investigation  

\Paragraph{Handling Function Calls}
By default, \coverme injects $\myr = \pen_i$ only in the entry
function to test. If the entry function invokes other external
functions, they will not be transformed. For example, in the program
\FOO of Fig.~\ref{fig:algo:injecting_pen}, we do not transform {\tt
  square(x)}. In this way, \coverme only attempts to saturate all
branches for a single function at a time.

However, \coverme can also easily handle functions invoked by its entry
function. As a simple example, consider:
\begin{lstlisting}
void FOO (double x) { GOO(x); }
void GOO (double x) { if (sin(x) <= 0.99) ... }
\end{lstlisting}

If \coverme aims to saturate {\tt FOO} and {\tt GOO} but not {\tt sin},
and  it sets \FOO as the entry function, then it instruments
both {\tt FOO} and {\tt GOO}. Only {\tt GOO} has a
conditional statement, and \coverme injects an assignment on {\tt
  r} in {\tt GOO}.

\section{Evaluation}
\label{sect:eval}
This section presents our extensive evaluation of CoverMe.  All
experiments are performed on a laptop with a 2.6 GHz Intel Core i7
running a Ubuntu 14.04 virtual machine with 4GB RAM.  The main results are presented in Tab.~\ref{tab:eval:me_vs_random}, ~\ref{tab:eval:me_vs_austin} and Fig.~\ref{fig:histograph:random}.

\begin{figure*}[htp]
\centering
\includegraphics[width=1.0\linewidth]{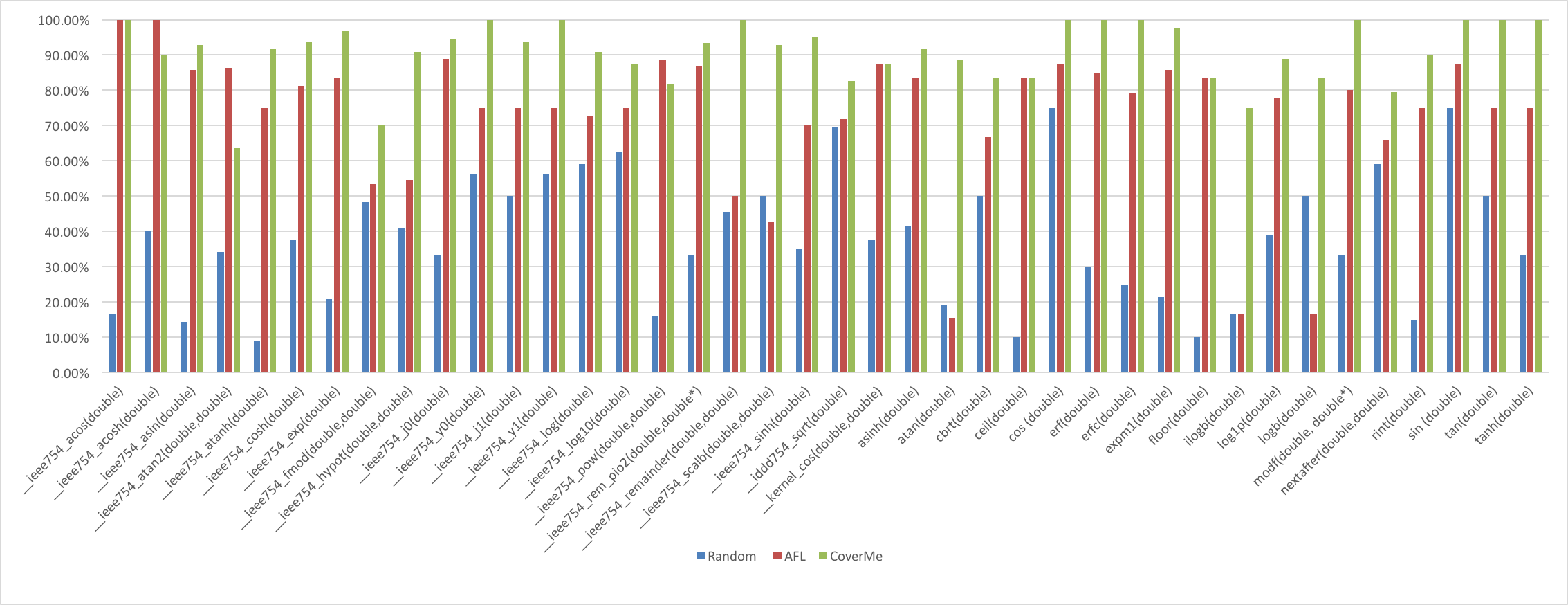}
\caption{Benchmark results corresponding to the data given in
  Tab.Tab.~\ref{tab:eval:me_vs_random}. The y-axis refers to the branch
  coverage; the x-axis refers to the benchmarks.}
\label{fig:histograph:random}
\end{figure*}

\subsection{Experimental Setup}
\label{subsect:expsetup}
%This section discusses the benchmarks, the tools for comparison, the settings and hardware, %and two evaluation objectives.

\Paragraph{Benchmarks}

Our benchmarks are C floating-point programs from the Freely
Distributable Math Library (Fdlibm) 5.3, developed by Sun
Microsystems, Inc.  These programs are available from the network
library netlib. We choose Fdlibm because it represents a set of
floating-point programs with reference quality and a large user group.
For example, the Java SE 8's math library is defined with respect to
Fdlibm 5.3.~\cite{javastrictmath:web}, Fdlibm is also ported to
Matlab, JavaScript, and has been integrated in Android.

Fdlibm 5.3 has 80 programs. Each program defines one or multiple math
functions.  In total, Fdlibm 5.3 contains 92 math functions.  Among
them, we exclude 36 functions that have no branch, 11 functions
involving input parameters that are not floating-point, and 5 static C
functions. Our benchmarks include \emph{all} remaining 40 functions.
Sect.~\ref{sect:untested} lists all excluded functions in Fdlibm 5.3.

%why they are not selected.

%To reproduce our
%experimental results, Fdlibm needs to be compile with its default
%option, CFlags ``\-D\_IEEE\_LIBM''.  

% We use the C math library Fdlibm 5.3~\cite{fdlibm:web} as our
% benchmarks. These programs are  developed by Sun Microsystems, Inc. They are used as the reference definition for Java Math Library, ported in Matlab and JavaScript, and have been integrated in Android. 

% rich in floating-point computations and branches, and have been
% used in Matlab, Java, JavaScript and Android.

% Fdlibm includes $80$ programs.  Each has one or multiple entry
% functions. In total, Fdlibm has $92$ entry functions. Among them, we
% exclude (1) $36$ functions that do not have  branches, (2) $11$ functions
% involving non-floating-point input parameters, and (3) $5$ static C
% functions. Our benchmark suite includes
% \emph{all} remaining $40$ functions in Fdlibm. 

%To reproduce our
%experimental results, Fdlibm needs to be compile with its default
%option, CFlags ``\-D\_IEEE\_LIBM''.  

%~\footnote{%
%  By default, C functions are \emph{[extern} and accessible from other
%  source files. However, a static C functions is only accessible from
%  the file where it is defined.}  

%To reproduce our results, Fdlibm needs to be compiled with the default build option of Fdlibm (which uses CFlags "-D_IEEE_LIBM'').

\Paragraph{Parameter Settings of CoverMe}
CoverMe supports three parameters: (1) the number of
Monte-Carlo iterations $\niter$, (2) the local optimization algorithm
\LM, and (3) the number of starting points $\nStart$.   They  correspond to  the three input parameters of 
Algo.~\ref{theory:algo:coverme}.  Our experiment sets
$\niter=5$, $\nStart=500$, and $\LM$ = ``powell'' (which refers  to
Powell’s  algorithm~\cite{Press:2007:NRE:1403886}).

\Paragraph{Tools for Comparison}
%Few coverage-based testing tools target floating-point code. 
We have compared CoverMe with three tools that support floating-point
coverage-based testing: 

\begin{itemize}
\item 
Rand is a pure random testing tool.  We  have implemented Rand using a pseudo-random number generator.

\item AFL~\cite{afl:web} is a gray-box testing tool released by
  the Google security team.  
It integrates a variety of guided search
 strategies and employs genetic algorithms to efficiently increase
 code coverage.

\item Austin~\cite{lakhotia2013austin} is a coverage-based testing
  tool that implements symbolic execution and search-based
  heuristics. Austin has been shown~\cite{lakhotia2010empirical} to be
  more effective than the testing tool called
  CUTE~\cite{Sen:2005:CCU:1081706.1081750} (which is not publicly available).

 \end{itemize}

 We have decided to not consider the following tools:
\begin{itemize}
\item Klee~\cite{Cadar:2008:KUA:1855741.1855756} is the
  state-of-the-art implementation of symbolic execution.  We do not
  consider Klee because its expression language does not support
  floating-point constraints. In addition, many operations in our
  benchmark programs, such as pointer reference, dereference, type
  casting, are not supported by Klee's backend solver
  STP~\cite{ganesh2007decision}, or any other backend solvers
  compatible with the Klee Multi-solver
  extension~\cite{palikareva2013multi}.\footnote{ In the recent
    Klee-dev mailing list, the Klee developers stated that Klee or
    its variant Klee-FP only has basic floating-point support and they
    are still "working on full FP support for Klee"
    \cite{kleedev:web}.}

\item Klee-FP~\cite{collingbourne2011symbolic} is a variant of Klee
  geared toward reasoning about floating-point value equivalence. It
  determines equivalence by checking whether two floating-point
  values are produced by the same operations~\cite{collingbourne2011symbolic}. We do not consider Klee-FP because its
  special-purpose design does not support coverage-based testing.

\item Pex~\cite{Tillmann:2008:PWB:1792786.1792798} is a coverage tool based on dynamic symbolic execution. We 
  do not consider Pex because it can only run for .NET programs on Windows
  whereas our tested programs are in C, and our testing platform is
  Linux.
\item FloPSy~\cite{Lakhotia:2010:FSF:1928028.1928039} is a
  floating-point testing tool that combines search based testing and
  symbolic execution.  We do not consider  this tool because it is developed by
  the same author of Austin and before Austin is released, and the
  tool is not available to us.

\end{itemize}

% \setlength{\tabcolsep}{2pt}
% \begin{table} \footnotesize
% 	\caption {Gcov metrics and explanations}\label{tab:algo:gcov}
% 	\begin{tabular}{p{0.625in} p{0.7in} p{1in} p{0.7in}}
% 		\toprule
% 		Metrics &  \mbox{Gcov} message & Description & Note \\
% 		\midrule
% 		Line\% & Lines executed & Covered source lines over total source lines & \mbox{\aka line or} \mbox{statement coverage} \\\arrayrulecolor{lightgray}\hline
% 		Condition\% & Branches \mbox{executed} & \mbox{Covered conditional} \mbox{statements over total} conditional statements \\\arrayrulecolor{lightgray}\hline
% 		Branch\% & Branches taken at least once & Covered branches over total branches & \aka \mbox{branch coverage} \\\arrayrulecolor{lightgray}\hline
% 		Call\% & Calls executed & \mbox{Covered calls over} \mbox{total} calls  \\
% 		\bottomrule
% 	\end{tabular}
% \end{table}

%\Paragraph{Evaluation Objective.}

\Paragraph{Coverage Measurement}
Our evaluation focuses on branch coverage. 
Sect.~\ref{sect:linecoverage} also shows our line coverage results.
For CoverMe and Rand, we use the Gnu coverage tool
Gcov~\cite{gcov:web}.  For AFL, we use AFL-cov~\cite{aflcov:web}, a
Gcov-based coverage analysis tool for AFL. For Austin, we calculate
the branch coverage by the number of covered branches provided in the
{\tt .csv} file produced by Austin when it terminates.

%Running time is \emph{not} our main evaluation objective, although we report the time whenever possible. 

\Paragraph{Time Measurement}
To compare the running time of CoverMe with the other tools requires
careful design.  CoverMe, Rand, and AFL all have the potentials to
achieve higher coverage if given more time or iteration cycles.
CoverMe terminates when it exhausts its iterations or achieves full
coverage, whereas Random testing and AFL do not terminate by
themselves.  In our experiment,  we first run CoverMe until
it terminates using the parameters provided above.  Then, we run Rand
and AFL with ten times of the CoverMe time.

Austin terminates when it decides that no more coverage can be
attained.  It reports no coverage results until it terminates its
calculation. Thus, it is not reasonable to set the same amount of time
for Austin as AFL and Rand.  The time for Austin refers to the time
Austin spends when it stops running.

% \toprule
% GCv Results (shown in Tab. X) & Gcov message & Meaning & Note \\
% \midrule
% Line \% & Lines executed & Covered source lines / Total source lines & Also known as "line/statement coverage" \\
% Condition \% & Branches executed & Covered conditional statements / Total conditional statement \\
% Branch \% & Branch taken at least once & Covered branches/ Total branches & Also known as "branch coverage" \\
% Call \% & Calls executed & Covered Calls / Total calls \\

% \bottomrule
% \end{tabular}
% \end{table*}

%%%%%%%%%%% say this later for each
%\paragraph{Assessement Tools}
%The coverage information is collected using the Gnu coverage, Gcov. 

\subsection{Quantitative Results}

%We present our comparison results between CoverMe, Rand, AFL, and
%Austin. As mentioned above, our evaluation objection is the branch
%coverage.

% This subsection presents two sets of experimental results. The first validates our approach by comparing it against random
% testing. %It is expected that our approach performs systematically
% %better than random testing.  
% The second compares CoverMe with Austin~\cite{lakhotia2013austin}, an open-source tool that combines symbolic execution and search-based strategies.
% ;

\subsubsection{CoverMe versus Random testing}

\setlength{\tabcolsep}{1pt}
\begin{table*}[p]\centering \footnotesize             %\footnotesize
  \caption {CoverMe versus  Rand and AFL.  The benchmark programs are taken from Fdlibm~\cite{fdlibm:web}. The coverage percentage is reported by  Gcov~\cite{gcov:web}. The time for CoverMe refers to the wall time. The time for Rand and AFL are set to be ten times of the CoverMe time. 
  } \label{tab:eval:me_vs_random}
  \begin{tabular}{llr|rrr|rrr|rr}
\toprule
    \multicolumn{3}{c|}{Benchmark}&\multicolumn{3}{c|}{Time (s)}&\multicolumn{3}{c|}{Branch  (\%)} & \multicolumn{2}{c|}{Improvement (\%)}  
		\\\arrayrulecolor{lightgray}\hline

File & Function & \#Branches & Rand & AFL & CoverMe & Rand  & AFL & CoverMe & CoverMe vs. Rand & CoverMe vs. AFL
		\\\arrayrulecolor{black}\midrule
e\_acos.c & ieee754\_acos(double) & 12 & 78 & 78 & 7.8 & 16.7 & 100.0 & 100.0 & 83.3 & 0.0  \\\arrayrulecolor{black}\midrule
e\_acosh.c & ieee754\_acosh(double) & 10 & 23 & 23 & 2.3 & 40.0 & 100.0 & 90.0 & 50.0 & -10.0  \\\arrayrulecolor{black}\midrule
e\_asin.c & ieee754\_asin(double) & 14 & 80 & 80 & 8.0 & 14.3 & 85.7 & 92.9 & 78.6 & 7.1  \\\arrayrulecolor{black}\midrule
e\_atan2.c & ieee754\_atan2(double,double) & 44 & 174 & 174 & 17.4 & 34.1 & 86.4 & 63.6 & 29.6 & -22.7  \\\arrayrulecolor{black}\midrule
e\_atanh.c & ieee754\_atanh(double) & 12 & 81 & 81 & 8.1 & 8.8 & 75.0 & 91.7 & 82.8 & 16.7  \\\arrayrulecolor{black}\midrule
e\_cosh.c & ieee754\_cosh(double) & 16 & 82 & 82 & 8.2 & 37.5 & 81.3 & 93.8 & 56.3 & 12.5  \\\arrayrulecolor{black}\midrule
e\_exp.c & ieee754\_exp(double) & 24 & 84 & 84 & 8.4 & 20.8 & 83.3 & 96.7 & 75.8 & 13.3  \\\arrayrulecolor{black}\midrule
e\_fmod.c & ieee754\_fmod(double,double) & 60 & 221 & 221 & 22.1 & 48.3 & 53.3 & 70.0 & 21.7 & 16.7  \\\arrayrulecolor{black}\midrule
e\_hypot.c & ieee754\_hypot(double,double) & 22 & 156 & 156 & 15.6 & 40.9 & 54.5 & 90.9 & 50.0 & 36.4  \\\arrayrulecolor{black}\midrule
e\_j0.c & ieee754\_j0(double) & 18 & 90 & 90 & 9.0 & 33.3 & 88.9 & 94.4 & 61.1 & 5.6  \\\arrayrulecolor{black}\midrule
 & ieee754\_y0(double) & 16 & 7 & 7 & 0.7 & 56.3 & 75.0 & 100.0 & 43.8 & 25.0  \\\arrayrulecolor{black}\midrule
e\_j1.c & ieee754\_j1(double) & 16 & 102 & 102 & 10.2 & 50.0 & 75.0 & 93.8 & 43.8 & 18.8  \\\arrayrulecolor{black}\midrule
 & ieee754\_y1(double) & 16 & 7 & 7 & 0.7 & 56.3 & 75.0 & 100.0 & 43.8 & 25.0  \\\arrayrulecolor{black}\midrule
e\_log.c & ieee754\_log(double) & 22 & 34 & 34 & 3.4 & 59.1 & 72.7 & 90.9 & 31.8 & 18.2  \\\arrayrulecolor{black}\midrule
e\_log10.c & ieee754\_log10(double) & 8 & 11 & 11 & 1.1 & 62.5 & 75.0 & 87.5 & 25.0 & 12.5  \\\arrayrulecolor{black}\midrule
e\_pow.c & ieee754\_pow(double,double) & 114 & 188 & 188 & 18.8 & 15.8 & 88.6 & 81.6 & 65.8 & -7.0  \\\arrayrulecolor{black}\midrule
e\_rem\_pio2.c & ieee754\_rem\_pio2(double,double*) & 30 & 11 & 11 & 1.1 & 33.3 & 86.7 & 93.3 & 60.0 & 6.7  \\\arrayrulecolor{black}\midrule
e\_remainder.c & ieee754\_remainder(double,double) & 22 & 22 & 22 & 2.2 & 45.5 & 50.0 & 100.0 & 54.6 & 50.0  \\\arrayrulecolor{black}\midrule
e\_scalb.c & ieee754\_scalb(double,double) & 14 & 85 & 85 & 8.5 & 50.0 & 42.9 & 92.9 & 42.9 & 50.0  \\\arrayrulecolor{black}\midrule
e\_sinh.c & ieee754\_sinh(double) & 20 & 6 & 6 & 0.6 & 35.0 & 70.0 & 95.0 & 60.0 & 25.0  \\\arrayrulecolor{black}\midrule
e\_sqrt.c & iddd754\_sqrt(double) & 46 & 156 & 156 & 15.6 & 69.6 & 71.7 & 82.6 & 13.0 & 10.9  \\\arrayrulecolor{black}\midrule
k\_cos.c & kernel\_cos(double,double) & 8 & 154 & 154 & 15.4 & 37.5 & 87.5 & 87.5 & 50.0 & 0.0  \\\arrayrulecolor{black}\midrule
s\_asinh.c & asinh(double) & 12 & 84 & 84 & 8.4 & 41.7 & 83.3 & 91.7 & 50.0 & 8.3  \\\arrayrulecolor{black}\midrule
s\_atan.c & atan(double) & 26 & 85 & 85 & 8.5 & 19.2 & 15.4 & 88.5 & 69.2 & 73.1  \\\arrayrulecolor{black}\midrule
s\_cbrt.c & cbrt(double) & 6 & 4 & 4 & 0.4 & 50.0 & 66.7 & 83.3 & 33.3 & 16.7  \\\arrayrulecolor{black}\midrule
s\_ceil.c & ceil(double) & 30 & 88 & 88 & 8.8 & 10.0 & 83.3 & 83.3 & 73.3 & 0.0  \\\arrayrulecolor{black}\midrule
s\_cos.c & cos (double) & 8 & 4 & 4 & 0.4 & 75.0 & 87.5 & 100.0 & 25.0 & 12.5  \\\arrayrulecolor{black}\midrule
s\_erf.c & erf(double) & 20 & 90 & 90 & 9.0 & 30.0 & 85.0 & 100.0 & 70.0 & 15.0  \\\arrayrulecolor{black}\midrule
 & erfc(double) & 24 & 1 & 1 & 0.1 & 25.0 & 79.2 & 100.0 & 75.0 & 20.8  \\\arrayrulecolor{black}\midrule
s\_expm1.c & expm1(double) & 42 & 11 & 11 & 1.1 & 21.4 & 85.7 & 97.6 & 76.2 & 11.9  \\\arrayrulecolor{black}\midrule
s\_floor.c & floor(double) & 30 & 101 & 101 & 10.1 & 10.0 & 83.3 & 83.3 & 73.3 & 0.0  \\\arrayrulecolor{black}\midrule
s\_ilogb.c & ilogb(double) & 12 & 83 & 83 & 8.3 & 16.7 & 16.7 & 75.0 & 58.3 & 58.3  \\\arrayrulecolor{black}\midrule
s\_log1p.c & log1p(double) & 36 & 99 & 99 & 9.9 & 38.9 & 77.8 & 88.9 & 50.0 & 11.1  \\\arrayrulecolor{black}\midrule
s\_logb.c & logb(double) & 6 & 3 & 3 & 0.3 & 50.0 & 16.7 & 83.3 & 33.3 & 66.7  \\\arrayrulecolor{black}\midrule
s\_modf.c & modf(double, double*) & 10 & 35 & 35 & 3.5 & 33.3 & 80.0 & 100.0 & 66.7 & 20.0  \\\arrayrulecolor{black}\midrule
s\_nextafter.c & nextafter(double,double) & 44 & 175 & 175 & 17.5 & 59.1 & 65.9 & 79.6 & 20.5 & 13.6  \\\arrayrulecolor{black}\midrule
s\_rint.c & rint(double) & 20 & 30 & 30 & 3.0 & 15.0 & 75.0 & 90.0 & 75.0 & 15.0  \\\arrayrulecolor{black}\midrule
s\_sin.c & sin (double) & 8 & 3 & 3 & 0.3 & 75.0 & 87.5 & 100.0 & 25.0 & 12.5  \\\arrayrulecolor{black}\midrule
s\_tan.c & tan(double) & 4 & 3 & 3 & 0.3 & 50.0 & 75.0 & 100.0 & 50.0 & 25.0  \\\arrayrulecolor{black}\midrule
s\_tanh.c & tanh(double) & 12 & 7 & 7 & 0.7 & 33.3 & 75.0 & 100.0 & 66.7 & 25.0 
\\\arrayrulecolor{black}\bottomrule
MEAN &  & 23 & 69 & 69 & 6.9 & 38.0 & 72.9 & 90.8 & 52.9 & 17.9 \\\arrayrulecolor{black}\bottomrule
\end{tabular}
\end{table*}

As a sanity check, we first  compare CoverMe with Rand.   In
Tab.~\ref{tab:eval:me_vs_random}, we sort all  programs
(Col.~1) and functions (Col.~2) by their names and give the numbers
of branches  (Col.~3).

Tab.~\ref{tab:eval:me_vs_random}, Col.~6 gives the time spent by
CoverMe. The time refers to the wall time reported by the Linux
command {\tt time}.  Observe that the time varies considerably,
ranging from 0.1 second ({\tt s\_erf.c}, {\tt erfc}) to 22.1
seconds ({\tt e\_fmod.c}). Besides, the time is not necessarily
correlated with the number of branches.  For example, CoverMe takes
1.1 seconds to run {\tt s\_expm1.c} (42 branches) and 10.1 seconds to
run {\tt s\_floor.c} (30 branches).  It shows the potential for
real-world program testing since CoverMe may not be very sensitive to
the number of lines or branches.

 Tab.~\ref{tab:eval:me_vs_random}, Col.~4
gives the time spent by Rand. Since Rand does not terminate by itself,
the time refers to the timeout bound. As mentioned above, we set the
bound as ten times of the CoverMe time.

% The numbers of
% branches (Col.~3) are are not necessarily  correlated with the
% running times. For example, CoverMe takes 1.1 seconds to run 
% {\tt s\_expm1.c} of 42 branches and 10.1 seconds to run
% the  {\tt s\_floor.c} of 30 branches. It shows the
% potential for real-world program testing since CoverMe may not be very
% sensitive to the number of lines or branches. 

%For information only, although we run Rand with the same amount of time as CoverMe, Rand could have been run with much less time. In our experiments, we find the 

Tab.~\ref{tab:eval:me_vs_random}, Col. 7 and 9 show the branch
coverage results of Rand and CoverMe respectively.  The coverage is
reported by the Gnu coverage tool Gcov~\cite{gcov:web}.  CoverMe
achieves 100\% coverage for 11 out of 40 tested functions with an
average of 90.8\% coverage, while Rand does not achieve any 100\%
coverage and attains only 38.0\% coverage on average.  The last row of
the table shows the mean values.  Observe that all values in Col.~9
are larger than the corresponding values in Col.~7. It means that
CoverMe achieves higher branch coverage than Rand for every benchmark
program. The result validates our sanity check.

Col.~10 is the improvement of CoverMe versus Rand. We calculate the
coverage improvement as the difference between their percentages.
CoverMe provides 52.9\% coverage improvement on average.

\begin{remark}
  Tab.~\ref{tab:eval:me_vs_random} shows that CoverMe achieves partial
  coverage for some tested programs. The incompleteness occurs in two
  situations: (1) The program under test has unreachable branches; (2)
  The representing function fails to confirm $\FOOR = 0$ when it in
  fact holds (Thm.~\ref{lem:overview:c2}). The latter can be caused by
  a weak optimization backend, which produces sub-optimal minimum
  points, or by floating-point inaccuracy when evaluating \FOOR.
  Sect.~\ref{sect:incompleteness}  illustrates these two situations.

\end{remark}

\subsubsection{CoverMe versus AFL}
Tab.~\ref{tab:eval:me_vs_random} also gives the experimental results
of AFL.  Col.~5 corresponds to the ``run time'' statistics provided by
the frontend of AFL. Col.~8 shows the branch coverage of AFL.  As
mentioned in Sect.~\ref{subsect:expsetup}, we terminate AFL once it
spends ten times of the CoverMe time.  Sect.~\ref{sect:aflsettings} 
gives additional details on our AFL settings.

Our results show that AFL achieves 100\% coverage for 2 out of 40
tested functions with an average of 72.9\% coverage, while CoverMe
achieves 100\% coverage for 11 out of 40 tested functions with an
average of 90.8\% coverage. The average improvement is 17.9\% (shown
in the last row).

%coverme outperforrm
The last column is the improvement of CoverMe versus AFL. We calculate
the coverage improvement as the difference between their percentages.
Observe that CoverMe outperforms AFL for 33 out of 40 tested
functions.  The largest improvement is 73.1\% with program {\tt
  s\_atan.c}.  For five of tested functions, CoverMe achieves the same
coverage as AFL. There are three tested functions where CoverMe
achieves less ({\tt e\_acosh.c}, {\tt e\_atan2.c}, and {\tt e\_pow.c}). 

\begin{remark}
  
  We have further studied CoverMe's coverage on these three programs
  versus that of AFL.  We have run AFL with the same amount of time as
  CoverMe (rather than ten times as much as CoverMe). With this
  setting, AFL does not achieves 70.0\% for {\tt e\_acosh.c}, 63.6\%
  for {\tt e\_atan2.c}, and 54.4\% for {\tt e\_pow.c}, which are less
  than or equal to CoverMe's coverage.

  % time of CoverMe, CoverMe achieves better results than AFL.  In our
  % investigation, we rerun AFL using the same amount of time as
  % CoverMe.  Then, AFL achieves 70.0\% for {\tt e\_acosh.c}, 63.6\% for
  % {\tt e\_atan2.c}, and 54.4\% for {\tt e\_pow.c}, whereas CoverMe
  % achieves (as shown in the table) 90.0\%, 63.6\%, and 81.6\%,
  % respectively.

  % That being said, we believe that comparing CoverMe and AFL by
  % running them with the same amount of time can be unfair, because AFL
  % usually requires long time to get good code coverage --- reason we
  % have set AFL to run ten times of CoverMe's time for the results
  % reported in Tab.~\ref{tab:eval:me_vs_random}.

  That being said, comparing CoverMe and AFL by running them using the
  same amount of time may be unfair because AFL usually requires much
  time to obtain good code coverage --- the reason why we have set AFL
  to run ten times as much time as CoverMe for the results reported in
  Tab.~\ref{tab:eval:me_vs_random}. 
\end{remark}

\subsubsection{CoverMe versus Austin}

\setlength{\tabcolsep}{3.5pt}
\begin{table*} \centering\footnotesize
	\caption{CoverMe versus Austin.  The benchmark programs are taken from the Fdlibm library~\cite{fdlibm:web}. The ``timeout''  refers to a time of more than 30000 seconds. The ``crash'' refers to a fatal exception when running Austin. }\label{tab:eval:me_vs_austin}
	\begin{tabular}{ll|rr|rr|rr}
		\toprule
		\multicolumn{2}{c|}{Benchmark}&\multicolumn{2}{c|}{Time (second)}&\multicolumn{2}{c|}{Branch coverage(\%)} &\multicolumn{2}{c}{Improvement}
		\\\arrayrulecolor{lightgray}\hline
		Program & Entry function  & Austin & CoverMe & Austin & CoverMe & Speedup &  Coverage (\%) \\\arrayrulecolor{black}\midrule
		e\_acos.c & ieee754\_acos(double) & 6058.8 & 7.8 & 16.7 & 100.0 & 776.4 & 83.3 \\\arrayrulecolor{black}\midrule
		e\_acosh.c & ieee754\_acosh(double) & 2016.4 & 2.3 & 40.0 & 90.0 & 887.5 & 50.0 \\\arrayrulecolor{black}\midrule
		e\_asin.c & ieee754\_asin(double) & 6935.6 & 8.0 & 14.3 & 92.9 & 867.0 & 78.6 \\\arrayrulecolor{black}\midrule
		e\_atan2.c & ieee754\_atan2(double, double) & 14456.0 & 17.4 & 34.1 & 63.6 & 831.2 & 29.6 \\\arrayrulecolor{black}\midrule
		e\_atanh.c & ieee754\_atanh(double) & 4033.8 & 8.1 & 8.3 & 91.7 & 495.4 & 83.3 \\\arrayrulecolor{black}\midrule
		e\_cosh.c & ieee754\_cosh(double) & 27334.5 & 8.2 & 37.5 & 93.8 & 3327.7 & 56.3 \\\arrayrulecolor{black}\midrule
		e\_exp.c & ieee754\_exp(double) & 2952.1 & 8.4 & 75.0 & 96.7 & 349.7 & 21.7 \\\arrayrulecolor{black}\midrule
		e\_fmod.c & ieee754\_frmod(double, double) & timeout & 22.1 & n/a & 70.0 & n/a & n/a \\\arrayrulecolor{black}\midrule
		e\_hypot.c & ieee754\_hypot(double, double) & 5456.8 & 15.6 & 36.4 & 90.9 & 350.9 & 54.6 \\\arrayrulecolor{black}\midrule
		e\_j0.c & ieee754\_j0(double) & 6973.0 & 9.0 & 33.3 & 94.4 & 776.5 & 61.1 \\\arrayrulecolor{black}\midrule
		& ieee754\_y0(double) & 5838.3 & 0.7 & 56.3 & 100.0 & 8243.5 & 43.8 \\\arrayrulecolor{black}\midrule
		e\_j1.c & ieee754\_j1(double) & 4131.6 & 10.2 & 50.0 & 93.8 & 403.9 & 43.8 \\\arrayrulecolor{black}\midrule
		& ieee754\_y1(double) & 5701.7 & 0.7 & 56.3 & 100.0 & 8411.0 & 43.8 \\\arrayrulecolor{black}\midrule
		e\_log.c & ieee754\_log(double) & 5109.0 & 3.4 & 59.1 & 90.9 & 1481.9 & 31.8 \\\arrayrulecolor{black}\midrule
		e\_log10.c & ieee754\_log10(double) & 1175.5 & 1.1 & 62.5 & 87.5 & 1061.3 & 25.0 \\\arrayrulecolor{black}\midrule
		e\_pow.c & ieee754\_pow(double, double) & timeout & 18.8 & n/a & 81.6 & n/a & n/a \\\arrayrulecolor{black}\midrule
		e\_rem\_pio2.c & ieee754\_rem\_pio2(double, double*) & timeout & 1.1 & n/a & 93.3 & n/a & n/a \\\arrayrulecolor{black}\midrule
		e\_remainder.c & ieee754\_remainder(double, double) & 4629.0 & 2.2 & 45.5 & 100.0 & 2146.5 & 54.6 \\\arrayrulecolor{black}\midrule
		e\_scalb.c & ieee754\_scalb(double, double) & 1989.8 & 8.5 & 57.1 & 92.9 & 233.8 & 35.7 \\\arrayrulecolor{black}\midrule
		e\_sinh.c & ieee754\_sinh(double) & 5534.8 & 0.6 & 35.0 & 95.0 & 9695.9 & 60.0 \\\arrayrulecolor{black}\midrule
		e\_sqrt.c & iddd754\_sqrt(double) & crash & 15.6 & n/a & 82.6 & n/a & n/a \\\arrayrulecolor{black}\midrule
		k\_cos.c & kernel\_cos(double, double) & 1885.1 & 15.4 & 37.5 & 87.5 & 122.6 & 50.0 \\\arrayrulecolor{black}\midrule
		s\_asinh.c & asinh(double) & 2439.1 & 8.4 & 41.7 & 91.7 & 290.8 & 50.0 \\\arrayrulecolor{black}\midrule
		s\_atan.c & atan(double) & 7584.7 & 8.5 & 26.9 & 88.5 & 890.6 & 61.6 \\\arrayrulecolor{black}\midrule
		s\_cbrt.c & cbrt(double) & 3583.4 & 0.4 & 50.0 & 83.3 & 9109.4 & 33.3 \\\arrayrulecolor{black}\midrule
		s\_ceil.c & ceil(double) & 7166.3 & 8.8 & 36.7 & 83.3 & 812.3 & 46.7 \\\arrayrulecolor{black}\midrule
		s\_cos.c & cos (double) & 669.4 & 0.4 & 75.0 & 100.0 & 1601.6 & 25.0 \\\arrayrulecolor{black}\midrule
		s\_erf.c & erf(double) & 28419.8 & 9.0 & 30.0 & 100.0 & 3166.8 & 70.0 \\\arrayrulecolor{black}\midrule
		& erfc(double) & 6611.8 & 0.1 & 25.0 & 100.0 & 62020.9 & 75.0 \\\arrayrulecolor{black}\midrule
		s\_expm1.c & expm1(double) & timeout & 1.1 & n/a & 97.6 & n/a & n/a \\\arrayrulecolor{black}\midrule
		s\_floor.c & floor(double) & 7620.6 & 10.1 & 36.7 & 83.3 & 757.8 & 46.7 \\\arrayrulecolor{black}\midrule
		s\_ilogb.c & ilogb(double) & 3654.7 & 8.3 & 16.7 & 75.0 & 438.7 & 58.3 \\\arrayrulecolor{black}\midrule
		s\_log1p.c & log1p(double) & 11913.7 & 9.9 & 61.1 & 88.9 & 1205.7 & 27.8 \\\arrayrulecolor{black}\midrule
		s\_logb.c & logb(double) & 1064.4 & 0.3 & 50.0 & 83.3 & 3131.8 & 33.3 \\\arrayrulecolor{black}\midrule
		s\_modf.c & modf(double, double*) & 1795.1 & 3.5 & 50.0 & 100.0 & 507.0 & 50.0 \\\arrayrulecolor{black}\midrule
		s\_nextafter.c & nextafter(double, double) & 7777.3 & 17.5 & 50.0 & 79.6 & 445.4 & 29.6 \\\arrayrulecolor{black}\midrule
		s\_rint.c & rint(double) & 5355.8 & 3.0 & 35.0 & 90.0 & 1808.3 & 55.0 \\\arrayrulecolor{black}\midrule
		s\_sin.c & sin (double) & 667.1 & 0.3 & 75.0 & 100.0 & 1951.4 & 25.0 \\\arrayrulecolor{black}\midrule
		s\_tan.c & tan(double) & 704.2 & 0.3 & 50.0 & 100.0 & 2701.9 & 50.0 \\\arrayrulecolor{black}\midrule
		s\_tanh.c & tanh(double) & 2805.5 & 0.7 & 33.3 & 100.0 & 4075.0 & 66.7 \\\arrayrulecolor{black}\bottomrule
		%\bottomrule
		MEAN &  & 6058.4 & 6.9 & 42.8 & 90.8 &3868.0 & 48.9\\
		\arrayrulecolor{black}\bottomrule
	\end{tabular}
\end{table*}

Tab.~\ref{tab:eval:me_vs_austin} compares the results of CoverMe and
Austin. We use the same set of benchmarks as
Tab.~\ref{tab:eval:me_vs_random} (Col.~1-2).  We use the time
(Col.~3-4) and the branch coverage metric (Col. 5-6) to evaluate the
efficiency and the coverage.  The branch coverage of Austin (Col.~5)
is provided by Austin itself rather than by Gcov.  Gcov needs to have
access to the generated test inputs to report the coverage, but Austin
does not provide a viable way to access the generated test inputs.

Austin shows large performance variances over different benchmarks,
from 667.1 seconds ({\tt s\_sin.c}) to hours. As shown in the last row
of Tab.~\ref{tab:eval:me_vs_austin}, Austin needs 6058.4 seconds on
average for the testing. The average time does not include the
benchmarks where Austin crashes\footnote{Austin raised an exception  when testing {\tt e\_sqrt.c}. The exception
  was triggered by {\tt AustinOcaml/symbolic/symbolic.ml} from
  Austin's code, at line 209, Column 1.} or times out. Compared with
Austin, CoverMe is faster (Tab.~\ref{tab:eval:me_vs_austin}, Col.~4)
with 6.9 seconds on average.

CoverMe achieves a higher branch coverage (90.8\%) than Austin (42.8\%).  We also compare across Tab.~\ref{tab:eval:me_vs_austin} and Tab.~\ref{tab:eval:me_vs_random}. On average, Austin provides slightly higher branch coverage (42.8\%) than Rand (38.0\%), but  lower than  AFL (72.9\%). 

Col.~7-8 are the improvement metrics of CoverMe against Austin. We calculate the 
Speedup (Col.~7) as the ratio of the time spent by
Austin and the time spent by CoverMe, and  the coverage improvement
(Col.~7) as the difference between the branch coverage
of CoverMe and that of Austin.  We observe that CoverMe provides
3,868X speed-up and 48.9\% coverage improvement on average.

% Although CoverMe provides significant improvement for the floating-point testing, it is important to note that CoverMe could be less
% effective for programs beyond floating-point logic. The representing function for integer programs is less smooth than
% what we have for floating-point programs. Thus, CoverMe
% complements, rather than competes with the state-of-the-arts, in
% particularly, those involving symbolic execution, such as Klee.

  \begin{remark}\label{remark:why}
    Three reasons contribute to CoverMe's effectiveness.  First,
    Thm.~\ref{lem:overview:c2} allows the search process to target the
    right test inputs only (each minimum point found in the
    minimization process corresponds to a new branch until all
    branches are saturated). Most random search techniques do not have
    such guarantee, so they may waste time searching for irrelevant
    test inputs. Second, SMT-based methods run into difficulties in
    certain kinds of programs, in particular, those with nonlinear
    arithmetic. As a result, it makes sense to use unconstrained
    programming in programs that are heavy on floating-point
    computation. In particular, the representing function for CoverMe
    is carefully designed to be smooth to some degree (\eg, the branch
    distances defined in Def.~\ref{def:overview:bd} are quadratic
    expressions), which allows CoverMe to leverage the power of local
    optimization and MCMC. Third, compared with symbolic execution
    based solutions, such as Austin, CoverMe invokes the unconstrained
    programming backend for minimizing a single representing function,
    whereas symbolic execution usually needs to invoke SMT solvers for
    solving a large number of paths.

    Since CoverMe has achieved high code coverage on most tested
    programs, one may wonder whether our generated inputs have
    triggered any latent bugs. Note that when no specifications are
    given, program crashes have frequently been used as an oracle for
    finding bugs in integer programs.  Floating-point programs, on the
    other hand, can silently produce wrong results without
    crashing. Thus, when testing floating-point programs, program
    crashes cannot be used as a simple, readily available oracle as
    for integer programs. Our experiments, therefore, have focused on
    assessing the effectiveness of CoverMe in solving the problem
    defined in Def.~\ref{def:overview:branchcoverage} and do not
    evaluate its effectiveness in finding bugs, which is orthogonal
    and exciting future work.

  \end{remark}

\section{Related Work}
\label{sect:relwork}
%[motivation]

Many survey papers~\cite{McMinn:2004:SST:1077276.1077279,sharma2014survey,utting2012taxonomy}
have reviewed the algorithms and implementations for coverage-based
testing.

% A:
% [major technique: symbolic reasoning]

%%%%%%%%s

%%%%%%%%%%%%original random testing part

% Random sampling is a generic solution for automated testing.  Because
% pure random testing is ineffective in finding deep semantic issues or
% handling large input space, many enhanced forms and heuristics have
% been proposed, \eg, adaptive random testing~\cite{chen2004adaptive}
% and the AFL fuzzer~\cite{afl:web}.

%As
%discussed in Remark~\ref{remark:why}, the major difference between
%CoverMe and random testing lies in that the former provides the
%theoretical guarantee (Thm.~\ref{lem:overview:c2}), making it a more
%effective in practice.

%%%%%%%%%%%%%added for the coverletter of the second phase 
\Paragraph{Random Testing} 
% Random sampling is a generic solution for automated testing.  Because
% pure random testing is ineffective in finding deep semantic issues or
% handling large input space, many enhanced forms and heuristics have
% been proposed, \eg, adaptive random testing~\cite{chen2004adaptive}
% and the AFL fuzzer~\cite{afl:web}.

The most generic automated testing solution may be to sample
from the input space randomly. Pure random testing is usually
ineffective if the input space is large, such as in the case of
testing floating-point programs.  Adaptive random
testing~\cite{chen2004adaptive} evenly spreads the randomly generated
test cases, which is motivated by the observation that neighboring
inputs often exhibit similar failure behavior.  AFL~\cite{afl:web} is
another improved random testing. Its sampling is aided with a set of
heuristics. AFL starts from a random seed and repeatedly mutates it
to inputs that can attain more program coverage.

%Despite the various forms of improvement,  random testing still has little chance of attaining XXX chances of attaining $-2$ or $2$ out of all  floats remains low.

% For example, random testing for program \FOO in Fig.~\ref{fig:algo:injecting_pen} may all falls in 1F, the opposite branch of {\texttt{square (x) == 4}}. 

% Random sampling is a generic solution for automated testing.  Because
% pure random testing is ineffective in finding deep semantic issues or
% handling large input space, many enhanced forms and heuristics have
% been proposed, \eg, adaptive random testing~\cite{chen2004adaptive}
% and the AFL fuzzer~\cite{afl:web}.

%%%%%%%%%%%%%%%%%
\Paragraph{Symbolic Execution}
   
Most branch coverage based testing algorithms follow the pattern of
symbolic execution~\cite{King:1976:SEP:360248.360252}. It selects a
target path $\tau$, derives a path condition $\Phi_{\tau}$, and calculates
a model of the path condition with an SMT solver.
The symbolic execution  approach is attractive because of its
theoretical guarantee, that is, each model of the path condition
$\Phi_{\tau}$ necessarily exercises its associated target path
$\tau$. In practice, however, symbolic execution can be ineffective if
there are too many target paths (\aka path explosion), or if the SMT
backend has difficulties in handling the path condition.  When
analyzing floating-point programs, symbolic execution and its DSE
(Dynamic Symbolic Execution) variants~\cite{conf/cav/SenA06,song2013bitblaze,DBLP:conf/kbse/BurnimS08}
typically reduce floating-point SMT solving to Boolean satisfiability
solving~\cite{Peleska:2011:ATC:1986308.1986333}, or approximate the
constraints over floats by those over
rationals~\cite{Lakhotia:2010:FSF:1928028.1928039} or
reals~\cite{DBLP:dblp_conf/popl/BarrVLS13}.
%, or rely on the theory
%combination framework
%~\cite{Nelson:1979:SCD:357073.357079,DBLP:journals/jacm/NieuwenhuisOT06}.

Fig.~\ref{fig:relwork:se_coverme} compares symbolic execution
approaches and CoverMe. Their major difference lies in that, while
symbolic execution solves a path condition $\Phi_{\tau}$ with an SMT
backend for each target path $\tau$, CoverMe minimizes a \emph{single}
representing function \FOOR with unconstrained programming. Because of
this difference, CoverMe does not have path explosion issues and in
addition, it does not need to analyze program semantics.  Similar to
symbolic execution, CoverMe also comes with a theoretical guarantee,
namely, each minimum that attains 0 must trigger a new branch,
guarantee that contributes to its effectiveness
(Sect. \ref{sect:eval}).

\begin{figure}[t]
\centering
\includegraphics[width=1.0\linewidth]{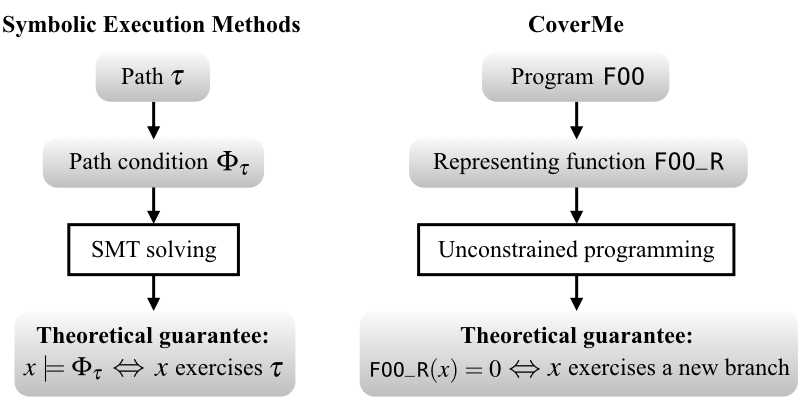}
\caption{Symbolic execution methods versus CoverMe.  }
\label{fig:relwork:se_coverme}
\end{figure}

\Paragraph{Search-based Testing}

Miller \etal~\cite{Miller:1976:AGF:1313320.1313530} probably are among
the first to introduce optimization methods in automated testing. They
reduce the problem of testing straight-line floating-point programs
into \emph{constrained programming}~\cite{Zoutendijk76,Minoux86}, that
is, optimization problems in which one needs to both minimize an objective
function and  satisfy a set of constraints.
Korel~\cite{Korel:1990:AST:101747.101755,Ferguson:1996:CAS:226155.226158}
extends Miller \etal's to general programs, leading to the development of
search-based testing~\cite{McMinn:2011:SST:2004685.2005495}. It views
a testing problem as a sequence of subgoals where each subgoal is
associated with a fitness function; the search process is then guided
by minimizing those fitness functions.  Search-based techniques have
been integrated into symbolic execution, such as
FloPSy~\cite{Lakhotia:2010:FSF:1928028.1928039} and
Austin~\cite{lakhotia2013austin}.

%  Their main issue is that they are
% not stand-alone algorithms, but rather ``strategies ready for adaption
% to specific problems'' as McMinn points out in his survey
% paper~\cite{McMinn:2004:SST:1077276.1077279}.

CoverMe's representing function is similar to the fitness function of
search-based testing in the sense that both reduce testing into
function minimization.  However, CoverMe uses the more efficient
unconstrained programming rather than the constrained programming used in
search-based testing. Moreover, CoverMe's use of representing function
comes up with a theoretical guarantee, which also opens the door to a
whole suite of optimization backends.

\section{Conclusion}
 \label{sect:conc}
 
We have introduced a new branch coverage based testing algorithm for
floating-point code.  We turn the challenge of testing floating-point
programs into the opportunity of applying unconstrained programming. Our core insight is to introduce the representing
function so that Thm.~\ref{lem:overview:c2} holds, which guarantees
that the minimum point of the representing function is an input that
exercises a new branch of the tested program.  We have implemented
this approach into the tool CoverMe. Our evaluation on Sun's math
library Fdlibm shows that CoverMe is highly effective, achieving
90.8\% of branch coverage in 6.9 seconds
on average, which largely outperforms random testing, AFL, and Austin.

For future work, we plan to investigate the potential synergies
 between CoverMe and symbolic execution, and extend this work to
 programs beyond floating-point code.

\Paragraph{Acknowledgment}
This research was supported in part by the United States National
Science Foundation (NSF) Grants 1319187, 1528133, and 1618158, and by
a Google Faculty Research Award. The information presented here does
not necessarily reflect the position or the policy of the Government
and no official endorsement should be inferred.  We would like to
thank our shepherd, Sasa Misailovic, and the anonymous reviewers for
valuable feedback on earlier drafts of this paper, which helped
improve its presentation.

\bibliographystyle{plain}
\bibliography{main}

\balance

%\clearpage

%%%%%%%%%%%%%below is to be uncommented for extended version.
\clearpage

 \appendix

\section{Untested programs in Fdlibm and Explanation}
 \label{sect:untested}
Tab.~\ref{tab:appendix:untested} lists all untested programs and
functions and explains the reasons. Three types
of the functions are excluded from the  evaluation.  They are (1)
functions without any branch, (2) functions input parameters that are not floating-point, and (3) static C functions.

\setlength{\tabcolsep}{6pt}
\begin{table*}\footnotesize
\caption{Untested programs  
in Fdlibm  and   the explanations.}
\centering
\begin{tabular}{l  l  |  l}
\toprule
File name & Function name &  Explanation \\\arrayrulecolor{black}\midrule
e\_gamma\_r.c & ieee754\_gamma\_r(double) & no branch \\\arrayrulecolor{lightgray}\hline
e\_gamma.c & ieee754\_gamma(double) & no branch \\\arrayrulecolor{lightgray}\hline
e\_j0.c & pzero(double) & static C function \\\arrayrulecolor{lightgray}\hline
 & qzero(double) & static C function \\\arrayrulecolor{lightgray}\hline
e\_j1.c & pone(double) & static C function \\\arrayrulecolor{lightgray}\hline
 & qone(double) & static C function \\\arrayrulecolor{lightgray}\hline
e\_jn.c & ieee754\_jn(int, double) & unsupported input type \\\arrayrulecolor{lightgray}\hline
 & ieee754\_yn(int, double) & unsupported input type \\\arrayrulecolor{lightgray}\hline
e\_lgamma\_r.c & sin\_pi(double) & static C function \\\arrayrulecolor{lightgray}\hline
 & ieee754\_lgammar\_r(double, int*) & unsupported input type \\\arrayrulecolor{lightgray}\hline
e\_lgamma.c & ieee754\_lgamma(double) & no branch \\\arrayrulecolor{lightgray}\hline
k\_rem\_pio2.c & kernel\_rem\_pio2(double*, double*, int, int, const int*) & unsupported input type \\\arrayrulecolor{lightgray}\hline
k\_sin.c & kernel\_sin(double, double, int) & unsupported input type \\\arrayrulecolor{lightgray}\hline
k\_standard.c & kernel\_standard(double, double, int) & unsupported input type \\\arrayrulecolor{lightgray}\hline
k\_tan.c & kernel\_tan(double, double, int) & unsupported input type \\\arrayrulecolor{lightgray}\hline
s\_copysign.c & copysign(double) & no branch \\\arrayrulecolor{lightgray}\hline
s\_fabs.c & fabs(double) & no branch \\\arrayrulecolor{lightgray}\hline
s\_finite.c & finite(double) & no branch \\\arrayrulecolor{lightgray}\hline
s\_frexp.c & frexp(double,  int*) & unsupported input type \\\arrayrulecolor{lightgray}\hline
s\_isnan.c & isnan(double) & no branch \\\arrayrulecolor{lightgray}\hline
s\_ldexp.c & ldexp(double, int) & unsupported input type \\\arrayrulecolor{lightgray}\hline
s\_lib\_version.c & lib\_versioin(double) & no branch \\\arrayrulecolor{lightgray}\hline
s\_matherr.c & matherr(struct exception*) & unsupported input type \\\arrayrulecolor{lightgray}\hline
s\_scalbn.c & scalbn(double, int) & unsupported input type \\\arrayrulecolor{lightgray}\hline
s\_signgam.c & signgam(double) & no branch \\\arrayrulecolor{lightgray}\hline
s\_significand.c & significand(double) & no branch \\\arrayrulecolor{lightgray}\hline
w\_acos.c & acos(double) & no branch \\\arrayrulecolor{lightgray}\hline
w\_acosh.c & acosh(double) & no branch \\\arrayrulecolor{lightgray}\hline
w\_asin.c & asin(double) & no branch \\\arrayrulecolor{lightgray}\hline
w\_atan2.c & atan2(double, double) & no branch \\\arrayrulecolor{lightgray}\hline
w\_atanh.c & atanh(double) & no branch \\\arrayrulecolor{lightgray}\hline
w\_cosh.c & cosh(double) & no branch \\\arrayrulecolor{lightgray}\hline
w\_exp.c & exp(double) & no branch \\\arrayrulecolor{lightgray}\hline
w\_fmod.c & fmod(double, double) & no branch \\\arrayrulecolor{lightgray}\hline
w\_gamma\_r.c & gamma\_r(double, int*) & no branch \\\arrayrulecolor{lightgray}\hline
w\_gamma.c & gamma(double, int*) & no branch \\\arrayrulecolor{lightgray}\hline
w\_hypot.c & hypot(double, double) & no branch \\\arrayrulecolor{lightgray}\hline
w\_j0.c & j0(double) & no branch \\\arrayrulecolor{lightgray}\hline
 & y0(double) & no branch \\\arrayrulecolor{lightgray}\hline
w\_j1.c & j1(double) & no branch \\\arrayrulecolor{lightgray}\hline
 & y1(double) & no branch \\\arrayrulecolor{lightgray}\hline
w\_jn.c & jn(double) & no branch \\\arrayrulecolor{lightgray}\hline
 & yn(double) & no branch \\\arrayrulecolor{lightgray}\hline
w\_lgamma\_r.c & lgamma\_r(double, int*) & no branch \\\arrayrulecolor{lightgray}\hline
w\_lgamma.c & lgamma(double) & no branch \\\arrayrulecolor{lightgray}\hline
w\_log.c & log(double) & no branch \\\arrayrulecolor{lightgray}\hline
w\_log10.c & log10(double) & no branch \\\arrayrulecolor{lightgray}\hline
w\_pow.c & pow(double, double) & no branch \\\arrayrulecolor{lightgray}\hline
w\_remainder.c & remainder(double, double) & no branch \\\arrayrulecolor{lightgray}\hline
w\_scalb.c & scalb(double, double) & no branch \\\arrayrulecolor{lightgray}\hline
w\_sinh.c & sinh(double) & no branch \\\arrayrulecolor{lightgray}\hline           
w\_sqrt.c & sqrt(double) & no branch \\\arrayrulecolor{black}\bottomrule
\end{tabular}
\label{tab:appendix:untested}
\end{table*}

\section{AFL Settings}
\label{sect:aflsettings} 
We run AFL following the instructions from its {\tt readme}
file~\footnote{\url{http://lcamtuf.coredump.cx/afl/README.txt}}: We first 
instrument the source program  using {\tt afl-gcc}, \footnote{We use
  AFL-clang, the alternative of {\tt afl-gcc}, for programs {\tt e\_j0.c} and
  {\tt e\_j1.c}.  {\tt afl-gcc} crashes on the two due a linking issue; we are
  unsure whether this is due to bugs in {\tt afl-gcc} or our misusing it.}
then we generate the test data using  {\tt afl-fuzz}. We terminate {\tt afl-fuzz} when it spends ten times of the CoverMe time.

%%%put this part for 
We show our  test harness programs when we run AFL on Fdlibm
programs.  Test harness for Fdlibm programs that accept a single {\tt
  double} parameter, \eg, {\tt e\_acos.c}: 

\lstset{xleftmargin=0.5cm, numbers=none}
\begin{lstlisting}
#include <stdio.h>
//LCOV_EXCL_START
int main(int argc, char *argv[])
{
  double x = 0;
  scanf("%lf", &x);
  fdlibm_program(x);
  return 0;
}
//LCOV_EXCL_STOP 
\end{lstlisting}

Test harness for Fdlibm programs  that accept two  {\tt  double} parameters, \eg, {\tt e\_atan2.c}:

\begin{lstlisting}
#include <stdio.h>
//LCOV_EXCL_START
int main(int argc, char *argv[])
{
  double x = 0;
  double y = 0;
  scanf("%lf %lf", &x, &y);
  fdlibm_program(x,y);
  return 0;
}
//LCOV_EXCL_STOP 
\end{lstlisting}

Test harness for Fdlibm programs that accept a {\tt  double} and a {\tt double*} parameters, \eg, {\tt e\_rem\_pio2.c}:

\begin{lstlisting}
#include <stdio.h>
//LCOV_EXCL_START
int main(int argc, char *argv[])
{
  double x = 0;
  double y = 0;
  scanf("%lf %lf", &x, &y);
  fdlibm_program(x,&y);
  return 0;
}
//LCOV_EXCL_STOP 
\end{lstlisting}

Above,  the {\tt LCOV\_EXCL\_START} and {\tt LCOV\_EXCL\_STOP} parts 
mark the beginning and the end of an excluded section. It allows 
the coverage analysis tool GCov~\cite{gcov:web} to ignore the lines in
between when reporting the coverage data.

%%% Local Variables:
%%% mode: latex
%%% TeX-master: t
%%% End:

\section{Line Coverage Results}
\label{sect:linecoverage}
The algorithm presented in the paper focuses on branch
coverage. CoverMe can also achieve high line coverage.  The line
coverage refers to the percentage of source lines covered by the
generated test data.  Tab.~\ref{tab:eval:line} provides the line
coverage results of CoverMe, AFL, and Rand.  The line coverage of
CoverMe, AFL, and Rand are shown in Col. 7, 8, 9 respectively. The
time is the same as in Tab.~\ref{tab:eval:me_vs_random}.  In
particular, CoverMe has achieved 100\% line coverage for 28 out of 40
tested functions. On average, CoverMe, AFL, and Rand get line coverage
of 97.0\%, 87.0\%, and 54.2\%, respectively.

\setlength{\tabcolsep}{1pt}
\begin{table*}[p]\centering \footnotesize             %\footnotesize
  \caption {Line Coverage: CoverMe versus Rand and AFL.  The benchmark programs are taken from Fdlibm~\cite{fdlibm:web}. We calculate the coverage percentage using Gcov~\cite{gcov:web}.
  } \label{tab:eval:line}
  \begin{tabular}{llr|rrr|rrr|rr}
\toprule
    \multicolumn{3}{c|}{Benchmark}&\multicolumn{3}{c|}{Time (s)}&\multicolumn{3}{c|}{Line  (\%)} & \multicolumn{2}{c|}{improvement (\%)}  
		\\\arrayrulecolor{lightgray}\hline
File & Function & \#Lines & Rand  & AFL & CoverMe & Rand   & AFL & CoverMe & CoverMe vs. Rand & CoverMe vs. AFL
		\\\arrayrulecolor{black}\midrule
e\_acos.c & ieee754\_acos(double) & 33 & 78 & 78 & 7.8 & 18.2 & 100.0 & 100.0 & 81.8 & 0.0 		\\\arrayrulecolor{black}\midrule
e\_acosh.c & ieee754\_acosh(double) & 15 & 23 & 23 & 2.3 & 46.7 & 100.0 & 93.3 & 46.7 & -6.7 		\\\arrayrulecolor{black}\midrule
e\_asin.c & ieee754\_asin(double) & 31 & 80 & 80 & 8.0 & 19.4 & 96.8 & 100.0 & 80.7 & 3.2 		\\\arrayrulecolor{black}\midrule
e\_atan2.c & ieee754\_atan2(double,double) & 39 & 174 & 174 & 17.4 & 59.0 & 97.4 & 79.5 & 20.5 & -17.9 		\\\arrayrulecolor{black}\midrule
e\_atanh.c & ieee754\_atanh(double) & 15 & 81 & 81 & 8.1 & 40.0 & 86.7 & 100.0 & 60.0 & 13.3 		\\\arrayrulecolor{black}\midrule
e\_cosh.c & ieee754\_cosh(double) & 20 & 82 & 82 & 8.2 & 50.0 & 95.0 & 100.0 & 50.0 & 5.0 		\\\arrayrulecolor{black}\midrule
e\_exp.c & ieee754\_exp(double) & 31 & 84 & 84 & 8.4 & 25.8 & 93.5 & 96.8 & 71.0 & 3.2 		\\\arrayrulecolor{black}\midrule
e\_fmod.c & ieee754\_fmod(double,double) & 70 & 221 & 221 & 22.1 & 54.3 & 60.0 & 77.1 & 22.9 & 17.1 		\\\arrayrulecolor{black}\midrule
e\_hypot.c & ieee754\_hypot(double,double) & 50 & 156 & 156 & 15.6 & 66.0 & 58.0 & 100.0 & 34.0 & 42.0 		\\\arrayrulecolor{black}\midrule
e\_j0.c & ieee754\_j0(double) & 29 & 90 & 90 & 9.0 & 55.2 & 100.0 & 100.0 & 44.8 & 0.0 		\\\arrayrulecolor{black}\midrule
 & ieee754\_y0(double) & 26 & 7 & 7 & 0.7 & 69.2 & 96.2 & 100.0 & 30.8 & 3.8 		\\\arrayrulecolor{black}\midrule
e\_j1.c & ieee754\_j1(double) & 26 & 102 & 102 & 10.2 & 65.4 & 100.0 & 100.0 & 34.6 & 0.0 		\\\arrayrulecolor{black}\midrule
 & ieee754\_y1(double) & 26 & 7 & 7 & 0.7 & 69.2 & 96.2 & 100.0 & 30.8 & 3.8 		\\\arrayrulecolor{black}\midrule
e\_log.c & ieee754\_log(double) & 39 & 34 & 34 & 3.4 & 87.7 & 87.2 & 100.0 & 12.3 & 12.8 		\\\arrayrulecolor{black}\midrule
e\_log10.c & ieee754\_log10(double) & 18 & 11 & 11 & 1.1 & 83.3 & 88.9 & 100.0 & 16.7 & 11.1 		\\\arrayrulecolor{black}\midrule
e\_pow.c & ieee754\_pow(double,double) & 139 & 188 & 188 & 18.8 & 15.8 & 92.1 & 92.7 & 76.9 & 0.6 		\\\arrayrulecolor{black}\midrule
e\_rem\_pio2.c & ieee754\_rem\_pio2(double,double*) & 64 & 11 & 11 & 1.1 & 29.7 & 82.8 & 92.2 & 62.5 & 9.4 		\\\arrayrulecolor{black}\midrule
e\_remainder.c & ieee754\_remainder(double,double) & 27 & 22 & 22 & 2.2 & 77.8 & 77.8 & 100.0 & 22.2 & 22.2 		\\\arrayrulecolor{black}\midrule
e\_scalb.c & ieee754\_scalb(double,double) & 9 & 85 & 85 & 8.5 & 66.7 & 77.8 & 100.0 & 33.3 & 22.2 		\\\arrayrulecolor{black}\midrule
e\_sinh.c & ieee754\_sinh(double) & 19 & 6 & 6 & 0.6 & 57.9 & 84.2 & 100.0 & 42.1 & 15.8 		\\\arrayrulecolor{black}\midrule
e\_sqrt.c & iddd754\_sqrt(double) & 68 & 156 & 156 & 15.6 & 85.3 & 85.3 & 94.1 & 8.8 & 8.8 		\\\arrayrulecolor{black}\midrule
k\_cos.c & kernel\_cos(double,double) & 15 & 154 & 154 & 15.4 & 73.3 & 100.0 & 100.0 & 26.7 & 0.0 		\\\arrayrulecolor{black}\midrule
s\_asinh.c & asinh(double) & 14 & 84 & 84 & 8.4 & 57.1 & 100.0 & 100.0 & 42.9 & 0.0 		\\\arrayrulecolor{black}\midrule
s\_atan.c & atan(double) & 28 & 85 & 85 & 8.5 & 25.0 & 25.0 & 96.4 & 71.4 & 71.4 		\\\arrayrulecolor{black}\midrule
s\_cbrt.c & cbrt(double) & 24 & 4 & 4 & 0.4 & 87.5 & 91.7 & 91.7 & 4.2 & 0.0 		\\\arrayrulecolor{black}\midrule
s\_ceil.c & ceil(double) & 29 & 88 & 88 & 8.8 & 27.6 & 100.0 & 100.0 & 72.4 & 0.0 		\\\arrayrulecolor{black}\midrule
s\_cos.c & cos (double) & 12 & 4 & 4 & 0.4 & 100.0 & 100.0 & 100.0 & 0.0 & 0.0 		\\\arrayrulecolor{black}\midrule
s\_erf.c & erf(double) & 38 & 90 & 90 & 9.0 & 21.1 & 78.9 & 100.0 & 79.0 & 21.1 		\\\arrayrulecolor{black}\midrule
 & erfc(double) & 43 & 1 & 1 & 0.1 & 18.6 & 95.3 & 100.0 & 81.4 & 4.7 		\\\arrayrulecolor{black}\midrule
s\_expm1.c & expm1(double) & 56 & 11 & 11 & 1.1 & 21.4 & 94.6 & 100.0 & 78.6 & 5.4 		\\\arrayrulecolor{black}\midrule
s\_floor.c & floor(double) & 30 & 101 & 101 & 10.1 & 26.7 & 100.0 & 100.0 & 73.3 & 0.0 		\\\arrayrulecolor{black}\midrule
s\_ilogb.c & ilogb(double) & 12 & 83 & 83 & 8.3 & 33.3 & 50.0 & 91.7 & 58.3 & 41.7 		\\\arrayrulecolor{black}\midrule
s\_log1p.c & log1p(double) & 46 & 99 & 99 & 9.9 & 71.7 & 91.3 & 100.0 & 28.3 & 8.7 		\\\arrayrulecolor{black}\midrule
s\_logb.c & logb(double) & 8 & 3 & 3 & 0.3 & 87.5 & 50.0 & 87.5 & 0.0 & 37.5 		\\\arrayrulecolor{black}\midrule
s\_modf.c & modf(double, double*) & 32 & 35 & 35 & 3.5 & 31.2 & 78.1 & 100.0 & 68.8 & 21.9 		\\\arrayrulecolor{black}\midrule
s\_nextafter.c & nextafter(double,double) & 36 & 175 & 175 & 17.5 & 72.2 & 88.9 & 88.9 & 16.7 & 0.0 		\\\arrayrulecolor{black}\midrule
s\_rint.c & rint(double) & 34 & 30 & 30 & 3.0 & 26.5 & 94.1 & 100.0 & 73.5 & 5.9 		\\\arrayrulecolor{black}\midrule
s\_sin.c & sin (double) & 12 & 3 & 3 & 0.3 & 100.0 & 100.0 & 100.0 & 0.0 & 0.0 		\\\arrayrulecolor{black}\midrule
s\_tan.c & tan(double) & 8 & 3 & 3 & 0.3 & 100.0 & 100.0 & 100.0 & 0.0 & 0.0 		\\\arrayrulecolor{black}\midrule
s\_tanh.c & tanh(double) & 16 & 7 & 7 & 0.7 & 43.8 & 87.5 & 100.0 & 56.3 & 12.5 		\\\arrayrulecolor{black}\midrule
MEAN &  & 32 & 69 & 69 & 6.9 & 54.2 & 87.0 & 97.0 & 42.9 & 10.0 
\\\arrayrulecolor{black}\bottomrule
\end{tabular}
\end{table*}

%%% Local Variables:
%%% mode: latex
%%% TeX-master: t
%%% End:

\clearpage
\section{Illustration of  the Incompleteness}
\label{sect:incompleteness}

This section discusses two benchmark programs where CoverMe did not attain full coverage. The first one is related to infeasible branches; the second one is due to the weakness of the unconstrained programming backend Basinhopping in sampling subnormal numbers.

\Paragraph{Missed branches in \texttt{k\_cos.c}} 

Fig.~\ref{fig:notfullcoverage:kcos} lists one of the smallest programs in our
tested benchmarks, {\tt k\_cos.c}. The program operates on two {\tt
  double} inputs. It first takes {\tt |x|}'s high word by bit
twiddling, including a bitwise {\tt AND} ($\&$), a point reference
(\&) and a dereference (*) operator. The bit twiddling result is then
stored in an integer variable {\tt ix} (Line 3), following which four
conditional statements depend on {\tt ix} (Lines 4-15).  As shown in
Tab.~\ref{tab:eval:me_vs_random}, CoverMe attained $87.5\%$ branch coverage.

By investigating the reason of the partial coverage, we found that the
program had one out of the eight branches that could not be reached:
In Line 5, the condition {\tt if ((int) x) == 0)} always holds because
it is nested within the $|x|<=2^{-27}$ branch. (We presume that Sun's
developers decided to use this statement to trigger the floating-point
\emph{inexact exception} as a side effect.) Since no inputs of {\tt
  \_\_kernel\_cos} can trigger the opposite branch of {\tt if (((int)
  x) == 0)}, the $87.5\%$ branch coverage is in fact optimal.

\Paragraph{Missed branches in \texttt{e\_fmod.c}} 
Fig.~\ref{fig:notfullcoverage:fmod} lists a segment of program {\tt
  e\_fmod.c} from line 57 to 72.  As shown in
Tab. ~\ref{tab:eval:me_vs_austin},  CoverMe covered 28 out of the 44 branches in the program.  

The comment is copied from the original program directly.  Line 57
(resp. Line 66) is a condition that can be triggered only if its input
$x$ (resp. $y$) is subnormal.  With our current settings specified in
Sect.~\ref{subsect:expsetup}, CoverMe could not generate any subnormal
number because its optimization backend could not.  Thus, CoverMe
missed the two branches in Lines 57 and 66 and all the nested
branches, which are 12 in total.

\begin{figure}[htp]
\lstset{xleftmargin=0.5cm, numbers=left}
\begin{lstlisting}
#define __HI(x) *(1+(int*)&x)
double __kernel_cos(double x, double y){
  ix = __HI(x)&0x7fffffff;  /* ix = |x|'s high word */
  if(ix<0x3e400000) {       /* if |x| < 2**(-27) */
     if(((int)x)==0) return ...; /* generate inexact */
  }
  ...;
  if(ix < 0x3FD33333) 	  /* if |x| < 0.3 */ 
    return ...;
  else {
    if(ix > 0x3fe90000) { /* if |x| > 0.78125 */
      ...;
    } else {
      ...;
    }
    return ...;
  }
}
\end{lstlisting}
\caption{Benchmark program   {\tt k\_cos.c}.}\label{fig:notfullcoverage:kcos}
\end{figure}

\begin{figure}[htp]
\lstset{xleftmargin=0.5cm, numbers=left}
\begin{lstlisting}[firstnumber=57]
if(hx<0x00100000) {	/* subnormal x */
  if(hx==0) {
    for (ix = -1043, i=lx; i>0; i<<=1) ix -=1;
  } else {
    for (ix = -1022,i=(hx<<11); i>0; i<<=1) ix -=1;
  }
} else ix = (hx>>20)-1023;

/* determine iy = ilogb(y) */
if(hy<0x00100000) {	/* subnormal y */
  if(hy==0) {
    for (iy = -1043, i=ly; i>0; i<<=1) iy -=1;
  } else {
    for (iy = -1022,i=(hy<<11); i>0; i<<=1) iy -=1;
  }
} else iy = (hy>>20)-1023;
\end{lstlisting}
\caption{Benchmark program {\tt e\_fmod.c}, Lines 57-72.}\label{fig:notfullcoverage:fmod} 
\end{figure}

%%% Local Variables:
%%% mode: latex
%%% TeX-master: t
%%% End:

\end{document}